\documentclass[11pt]{amsart}
\usepackage{amsmath,amssymb,amsthm,tikz,bbm,bm}

\RequirePackage{color}
\RequirePackage[colorlinks, urlcolor=my-blue,linkcolor=my-red,citecolor=my-green]{hyperref}
\definecolor{my-blue}{rgb}{0.0,0.0,0.6}
\definecolor{my-red}{rgb}{0.5,0.0,0.0}
\definecolor{my-green}{rgb}{0.0,0.5,0.0}
\definecolor{nicos-red}{rgb}{0.75,0.0,0.0}
\definecolor{light-gray}{gray}{0.6}
\definecolor{really-light-gray}{gray}{0.8}
\definecolor{sussexg}{rgb}{0.0,0.5,0.5}
\definecolor{sussexp}{rgb}{0.5,0.0,0.5}

\usepackage{nicefrac}
\newtheorem{theorem}{\sc Theorem}[section]

\newtheorem{proposition}[theorem]{\sc Proposition}
\newtheorem{corollary}[theorem]{\sc Corollary}
\newtheorem{assumption}[theorem]{\bf Assumption}

\numberwithin{equation}{section}
\theoremstyle{remark}
\newtheorem{remark}[theorem]{Remark}
\newtheorem{example}{\bf Example}

\newcommand{\be}{\begin{equation}}
\newcommand{\ee}{\end{equation}}

\def\bE{\mathbb{E}}
\def\bN{\mathbb{N}}
\def\bP{\mathbb{P}}

\def\bZ{\mathbb{Z}}

 \def\Z{\bZ}

\def\N{\bN}

\def\E{\bE}
\def\P{\bP} %% environment measure 

\definecolor{darkgreen}{rgb}{0.0,0.5,0.0}
\definecolor{darkblue}{rgb}{0.0,0.0,0.3}
\definecolor{nicosred}{rgb}{0.65,0.1,0.1}
\definecolor{light-gray}{gray}{0.7}
\allowdisplaybreaks[1]

\RequirePackage{datetime} %comment out when arxiving or submitting

\begin{document}
\usdate
\title[Derivation of triple closures]
{Theoretical and numerical considerations of the assumptions behind triple closures in epidemic models on networks}

\author[N.~Georgiou]{Nicos Georgiou}
\address{Nicos Georgiou\\ University of Sussex\\ Department of  Mathematics \\ Falmer Campus\\ Brighton BN1 9QH\\ UK.}
\email{n.georgiou@sussex.ac.uk}
\urladdr{http://www.sussex.ac.uk/profiles/329373} 

\author[I. Z.~Kiss]{Istv\'an Z. Kiss}
\address{Istvan Kiss\\ University of Sussex\\ Department of  Mathematics \\ Falmer Campus\\ Brighton BN1 9QH\\ UK.}
\email{i.z.kiss@sussex.ac.uk}
\urladdr{http://www.sussex.ac.uk/profiles/211073}

\author[P.~Simon]{ P\'eter Simon}
\address{P\'eter Simon\\ Institute of Mathematics, E\"otv\"os Lor\'and University Budapest, and
Numerical Analysis and Large Networks Research Group, Hungarian Academy of Sciences\\ Hungary.}
\email{simonp@cs.elte.hu}
\urladdr{http://simonp.web.elte.hu/index\_en.html}

\thanks{NG was partially supported by the EPSRC First Grant EP/P021409/1. IZK was partially supported by the Leverhulme Trust, Research Project Grant RPG-2017-370. PS was partially supported by the Hungarian Scientific Research Fund, OTKA, Grant no. 115926.}



\date{\today}

\begin{abstract}
Networks are widely used to model the contact structure within a population and in the resulting models of disease spread.
While networks provide a high degree of realism, the analysis of the exact model is out of reach and
even numerical methods fail for modest network size. Hence, mean-field models (e.g. pairwise) focusing on describing the
evolution of some summary statistics from the exact model gained a lot of traction over the last few decades.
In this paper we revisit the problem of deriving triple closures for pairwise models and we investigate in detail the assumptions
behind some of the well-known closures as well as their validity. Using a top-down approach we start at the level of the entire graph and work down to the level
of triples and combine this with information around nodes and pairs. We use our approach to derive many of the existing closures
and propose new ones and theoretically connect the two well-studied models of multinomial link and Poisson link selection. 
The theoretical work is backed up by numerical examples to highlight where the commonly used assumptions may fail and
provide some recommendations for how to choose the most appropriate closure when using graphs with no or modest degree heterogeneity. 
\end{abstract}

\maketitle

\section{Introduction}
%%%%%%%%%%%%%%%%%%%%%%%%%%%%%%%%%%%%%%%%

Many complex systems are forbiddingly high dimensional and one efficient way to deal with such a challenge is to focus on some
summary statistics or moments of the full-system. However, more often than not, the evolution of a moment (usually a coarse grained view or quantity that is computable from the full
model) depends on higher-order moments. In~\cite{kuehn2016moment}, the author summarises the four main steps of a typical moment closure based approach. 
These are: (a) select the moments and their hierarchy, (b) write down evolution equations for the moments, (c) derive and apply the moment-closure and, finally, 
(d) justify and validate the choice of moments and closures and perform further numerical tests.

Unfortunately, there is no single or precise way in which moment hierarchies and moment closures are derived. Often these rely on empirical or numerical 
observations, may only work for particular problems or in  specific contexts and rigorous mathematical proofs for closures are rare. As a result in
this paper, we revisit some of the existing closures in the context of epidemic models and show that these can be derived by a 
hybrid method combining a top-down approach (counting over the entire graph/network) and a bottom-up approach where 
assumptions about the states of the neighbours of nodes are made.

This is a well-know and well-studied area with some initial important results by \cite{rand1999correlation,keeling1999effects}. In particular they showed that closures 
for some epidemic models are possible by starting at the level of nodes and considering their degrees as a random variable with a given mean plus 
some random error with zero mean. Coupling this with knowledge at the graph-level about the counts of nodes, links and triples in various states, 
they managed to derive two distinct closures which we will  investigate in detail. One of the main assumption made in their approach is about the distribution of the states of nodes around a given node  (i.e. Multinomial or Poisson).

In what follows we start from the entire graph and aim to re-derive some of the existing closures and propose two new ones, 
but with the main contribution being that we provide a rather general approach and we illustrate it by using a hypothetical (i.e. uniform) form of the 
distribution of the sates of the nodes around a give node.

%The main contribution being that we provide a rather general approach when thinking about closures.This process is referred to as `closure' and it is fundamental to derive mean-field %models. In what follows we start from the entire graph and aim to re-derive some of the existing closures and propose new ones, but with the main contribution being that we provide a %rather general approach when thinking about closures.\\

The paper is structured as follows. In Section~\ref{sec:model} we give the main ingredients of the model including the network and epidemic dynamics on it together with the unclosed  pairwise model for SIS dynamics. In Section~\ref{sec:prob} we present the top-down derivation of an exact expression for the expected value of $k$-tuples of states in the epidemics via probabilistic considerations. In Section \ref{sec:closures} we re-derive some of the classical closures used in the literature and present two new ones, that performs well in numerical test. We also highlight when the Multinomial and Poisson assumptions work and when they are hard to distinguish. 
The final part of the section explains how to rigorously derive the Poisson model from the Multinomial link model, both widely used in the literature. 
 In Section~\ref{sec:numerics_and_threshold} we provide a large number of numerical tests of the closures and show that the pairwise model with the new closure leads to an
epidemic threshold which is well-known in the literature. Further results and applications of our method are presented in Section~\ref{sec:further_extensions}.
Finally, Section~\ref{sec:disc} is dedicated to a discussion and summary of our findings.

%%%%%%%%%%%%%%%%%%%%%%%%%%%%%%%%%%%%%%%%
\section{Model: network, epidemic dynamics and mean-field models}\label{sec:model}
%%%%%%%%%%%%%%%%%%%%%%%%%%%%%%%%%%%%%%%%
The starting point is to model the contact structure of a population of $N$ individuals as an undirected  and unweighted network $\mathcal G=(\mathcal V, \mathcal E)$, so that $|\mathcal V| = N$. Such a graph can also be represented  by adjacency matrix $G=(g_{ij})_{i,j=1,2,\dots, N}$, where $g_{ij}=1$ if nodes $i$ and $j$ are connected and zero otherwise. Self-loops are excluded, so $g_{ii}=0$ and $g_{ij}=g_{ji}$ for all $i,j=1,2, \dots N$. While we will consider a general epidemic model where nodes can be in an arbitrary number of discrete states (i.e. $m$ different states $A_1, A_2, \ldots A_m$), our examples will focus on the standard susceptible-infected-susceptible (SIS) epidemic dynamics on a network. The SIS dynamics is driven by two processes: (a) infection and (b) recovery from
infection. Infection spreads from an infected node (I) to any of its susceptible neighbours (S) and this is modelled as a Poisson point process with per-link infection rate $\tau$.
Infectious nodes recover from infection at constant rate $\gamma$, independently of the network, and become susceptible again. The resulting model is a continuous-time Markov Chain
over a state space with $2^N$ elements. This consists of all arrangements of length $N$ with each entry being either S or I independently. While this is easy to generalise and write
down theoretically, the numerical integration of the system becomes intractable even for modest values of $N$~\cite{KissMillerSimon}.

One way to deal with such a high-dimensional model is to derive mean-field approximations for some of the summary statistics of the exact process.
Probably of most interest is the expected number of infected nodes over time. There are many different approaches that can help achieve this~\cite{KissMillerSimon},
but almost all rely on starting at `node' or `node and its neighbourhood-level' and proceed by writing down differential equations for their evolution. This immediately
leads to a dependency on higher-order moments since whatever quantity we focus on, its evolution will depend on the type and states of
the neighbouring nodes, e.g. singles depend on pairs and pairs depend on triple. One straightforward mean-field model, also extensively used in this paper, is the pairwise model which is given below,
\begin{align}
\dot{[S]}&=-\tau[SI]+\gamma[I], \label{eq:PW_SIS_S} \\
\dot{[I]}&=\tau[SI] -\gamma[I],\label{eq:PW_SIS_I} \\
\dot{[SS]}&=-2\tau[SSI]+2\gamma[SI], \label{eq:PW_SIS_SS} \\
\dot{[SI]}&=\tau([SSI]-[ISI]-[SI])-\gamma([SI]-[II]), \label{eq:PW_SIS_SI}\\
\dot{[II]}&=2\tau([ISI]+[SI])-2\gamma[II]. \label{eq:PW_SIS_II}
\end{align}
%dst(1)=-tau*SI+gam*I;
%dst(2)=tau*SI-gam*I;
%dst(3)=-(gam+tau)*SI+tau*(SSI-ISI)+gam*II;
%dst(4)=-2*tau*SSI+2*gam*SI;
%dst(5)=-2*gam*II+2*tau*(ISI+SI);
Here, $[\cdot]$ stands for the expected number of the respective quantities. 
%\textbf{I would leave out this paragraph. Furthermore, let $A_{i}$ equal 1 if the individual at node $i$ is of type $A$ and equal zero otherwise. Then
%single nodes (singles) of type $A$ can be counted as $[A]=\sum_{i}A_{i}$, pairs of nodes (pairs) of type $A-B$ can be counted as $[AB]=\sum_{i,j}A_{i}B_{j}g_{ij}$ and triples of nodes
%(triples) of type $A-B-C$ can be counted as $[ABC]=\sum_{i,j,k}A_{i}B_{j}C_{k}g_{ij}g_{jk}$. This method of counting means that pairs are counted once in each direction, so $%[AB]=[BA]
%$, and $[AA]$ is even.}
%\textbf{Instead I would say something like this: 
This unclosed model was derived from the full system of master equations and it was proved to be exact~\cite{Taylor2012JMB}.

The equations above are straightforward to interpret~\cite{KissMillerSimon}. The evolution equation $[I]$, see~\eqref{eq:PW_SIS_I}, has two terms: (a) a positive term which is proportional to the expected
number of S-I links ($[SI]$) and represents the incoming flux of new infections, and (b) a negative term which is proportional to $[I]$ and stands for the recovery of infected nodes. More
importantly, we notice that singles depend on pairs and this hierarchy of dependency continues whereby $[SI]$ links are created at rate $\tau[SSI]$, meaning that one of the S nodes in
an S-S link can be infected by an external node, e.g. S-S-I. Equally, the $[SI]$ links are depleted due to events within pair, i.e. I can infect S or I can recover, or events from outside the
pair, i.e. an external node infecting S, that is I-S-I. It is now clear that there is hierarchy of dependency on ever higher moments.

Keeping the hierarchy would lead to a large number of more and more complicated equations. But what if higher order moments, say triples, can be approximated by lower-order ones, such as singles and pairs? This process is referred to as `closure' and it is fundamental to derive mean-field models. In what follows we start from the entire graph and aim to re-derive some of the existing closures and propose new ones, but with special focus on staring at the graph-level and combining this with information at the local or node-level.

%\textbf{These two new references should be inserted:
%}

%%%%%%%%%%%%%%%%%%%%%%%%%%%%%%%%%%%%%%%%
%\section{Probabilistic approach}
\section{Triple counts from neighbourhood distribution}\label{sec:prob}
%%%%%%%%%%%%%%%%%%%%%%%%%%%%%%%%%%%%%%%%

%For the most general setting, consider a generic network   $\mathcal G=(\mathcal V, \mathcal E)$, so that $|\mathcal V| = N$ and assume epidemic dynamics are run with $m$ different types $A_1, A_2, \ldots A_m$. The Kolmogorov equations may require a closure of the correlations at any length,
Deriving evolution equation for some of the summary statistics from the exact model will require to find reasonable approximations to expressions of the form
\[
[A_{i_1}A_{i_2}\ldots A_{i_k}], \text{ where } k \in \N \text{ and } i_j \in \{ 1, \ldots, m\} \text{ for all } j.
\]
All quantities depend on time, so we assume it fixed and do not encumber notation with an extra $t$ index. Also keep $k$ fixed for the moment. Define the set of paths of size $k$ as
\[
\Pi_k = \{(v_1,  v_2 , \ldots, v_k): (v_1,  v_2 , \ldots, v_k) \in \mathcal V^k, v_i \neq v_j, (v_i, v_{i+1}) \in \mathcal E \text{ for } 1 \le i \le k-1 \}.
\]
Admissible and non-admissible paths are shown in Figure~\ref{fig:admissible_k_tuple}.
Irrespective of the Markov chain which happens on the network $\mathcal G$, the process will induce a probability mass function on the $k$-tuples of $\Pi_k$. 
At a given time every element of the state space has a well-defined probability mass which is given by the Markov chain itself, and obviously depends also on the initial state of the system. Then the probabilities of the states determine the expected value of the elements of $\Pi_k$.
%The set of $k$-tuples we study this measure is $\Pi_k$.  Expectations and probabilities below are about $\Pi_k$.
%\textbf{The notion of induced measure is a bit scary here. I would simply say that: at a given time every element of the state space has a well-defined probability which is given by the %Markov chain itself, and obviously depends also on the initial state of the system. Then the probabilities of the states determine the expected value of the elements of $\Pi_k$.}

\begin{figure}[h]
\begin{center}
\begin{tikzpicture}[scale=0.95]
\draw[magenta, line width = 1.5pt] (0,0)--(1,0)--(1,1)--(0,1)--(0,0)--(-1,-1)--(-1,2);
\draw[magenta, line width = 1.5pt](1,1)--(1,1.5)--(1,1)--(1.5,1.5);
\draw[magenta, line width = 1.5pt](-1,2)--(-1.5,2.5)--(-1,2)--(-0.5,2.5);
\draw[magenta, line width = 1.5pt](-1,-1)--(-1.5,-1.5);
\draw[line width = 3pt](-1,2)node[below left]{$v_1$}--(0,0)node[left]{$v_2$}--(1,0)node[below left]{$v_3$}--(2, -0.9)node[below left]{$v_4$}--(1,1)node[below left]{$v_5$};
\draw[magenta, line width = 1.5pt](2,-1)--(2.5,-1.5);
%%%%%%%%
\draw[fill=magenta](0,0)circle(1.5mm);
\draw[fill=magenta](0,1)circle(1.5mm);
\draw[fill=magenta](1,0)circle(1.5mm);
\draw[fill=magenta](-1,-1)circle(1.5mm);
\draw[fill=magenta](1,1)circle(1.5mm);
\draw[fill=magenta](2,-1)circle(1.5mm);
\draw[fill=magenta](-1,2)circle(1.5mm);
%%%%%%%%
\draw[magenta, line width = 1.5pt] (6,0)--(7,0)--(7,1)--(6,1)--(6,0)--(5,-1)--(5,2);
\draw[magenta, line width = 1.5pt](7,1)--(7,1.5)--(7,1)--(7.5,1.5);
\draw[magenta, line width = 1.5pt](5,2)--(4.5,2.5)--(5,2)--(5.5,2.5);
\draw[magenta, line width = 1.5pt](5,-1)--(4.5,-1.5);
\draw[magenta, line width = 1.5pt](8,-1)--(8.5,-1.5);
\draw[line width = 1.5pt, magenta](5,2)--(6,0)--(7,0)--(8, -1)--(7,1);
\draw[line width = 3pt](5,2)node[below left]{$v_1$}--(6,0)node[left]{$v_2$}--(7,0)node[below left]{$v_3$}--(6, 0)node[below]{$v_4$}--(5,-1)node[below right]{$v_5$};
%%%%%%%%
%%%%%%%%
\draw[fill=magenta](6,0)circle(1.5mm);
\draw[fill=magenta](6,1)circle(1.5mm);
\draw[fill=magenta](7,0)circle(1.5mm);
\draw[fill=magenta](5,-1)circle(1.5mm);
\draw[fill=magenta](7,1)circle(1.5mm);
\draw[fill=magenta](8,-1)circle(1.5mm);
\draw[fill=magenta](5,2)circle(1.5mm);
\end{tikzpicture}
\end{center}
\caption{Admissible (left) and non-admissible (right) ordered network paths going through 5 nodes. The left path visits each of its 5 nodes exactly once. The right one is inadmissible in our calculations, as the second and fourth node on it coincide.}
\label{fig:admissible_k_tuple}
\end{figure}
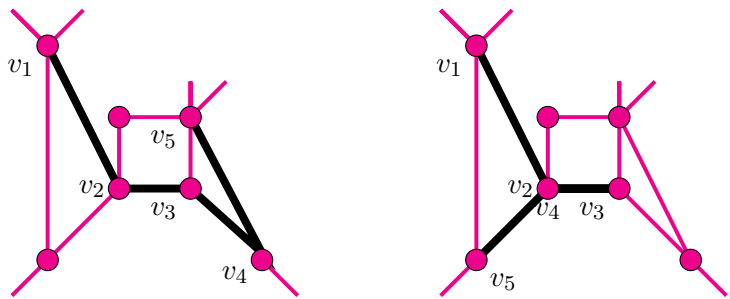

%\textbf{I would leave out this paragraph, since it was explained at the end of the previous section. For example, in and $SIS$ epidemic on a network knowing
%the expected number of $S-I$ links, i.e. $[SI]$, would lead to a very compact mean-field model with one single equation. The problem is that a formula which provides this quantity for an arbitrary graph and a choice of disease parameters does not exist, that is except for the fully connected network where $[SI]=[I](N-[I])$. Hence, one may choose to leave the pairs unclosed and move to the level of triples where one may be able to relate expected number of triples to expected numbers of pairs and triples.}

An expression of the expected value of the $k-$tuple can be computed using the law of total probability, by conditioning on the type of site $v_j$: 
\allowdisplaybreaks
\begin{align}
\E(A_{i_1}&A_{i_2}\ldots A_{i_k})
= \E\Big(\sum_{(v_1, \ldots, v_k) \in \Pi_k } \mathbbm1\{ v_1 = A_{i_1}, v_2 = A_{i_2}, \ldots, v_k=A_{i_k} \} \Big) \notag \\
&= \sum_{(v_1, \ldots, v_k) \in \Pi_k } \P \{ v_1 = A_{i_1}, v_2 = A_{i_2}, \ldots, v_k=A_{i_k} \} \label{eq:linky}\\
%&= \sum_{v \in \mathcal G} \sum_{ (v_1, \ldots, v_j = v, \ldots, v_k) \in \Pi_k} P \{ v_1 = A_{i_1},  \ldots, v_k=A_{i_k} | v_j = v = A_{i_j}\} \P\{ v = A_{i_j}\}  \notag\\
&= \sum_{v \in \mathcal G}   \P\{ v = A_{i_j}\} \!\!\!\! \sum_{ (v_1, \ldots, v_j = v, \ldots, v_k) \in \Pi_k}  \!\!\!\! P \{ v_1 = A_{i_1},  \ldots, v_k=A_{i_k} | v_j = v = A_{i_j}\}.\label{eq:condmass}
\end{align}

This expression is based on the fact that we are computing the expectation of a specific arrangement of states on a non-self-intersecting  paths. Different expressions involving conditional quantities for expected values of more complicated sets of $k$-tuples will hold in the case of non-admissible paths; the applications below are for $k= 3$ involving open triples and therefore we do not need to consider more complicated paths. 

For the theoretical part of this article, we consider a multi-type epidemic on $\mathcal G$ and assume a node can be in any of $m$ different states, say $A_1, \ldots, A_m$. We will first compute theoretical expressions for $[A_{i} A A_{j}]$ using \eqref{eq:condmass}.

For a given node $v$ denote by
\be
N_i^{(v)} = \text{card}\{ \text{neighbours of $v$ that take the value $A_i$} \}, \quad A_i \in \{ A_1, \ldots, A_m \}.
\ee
Keep in mind that $ \displaystyle \sum_{i=1}^m N_i^{(v)} = \deg v$.

\begin{proposition}\label{prop:triples}
Consider a multi-type epidemic on $\mathcal G$ and assume a node can be on any of $m$ different states $A_1, \ldots, A_m$. Then, the following two formulas hold:

\noindent For $A_i \neq A_j$
\be \label{eq:more-ex}
[A_iAA_j]  =  \sum_{v \in \mathcal G}    \P\{ v = A\} E \Big[ N_i^{(v)}N_j^{(v)} \Big| v= A\Big],
\ee
for $A_i = A_j$
\be \label{eq:more-ex2}
[A_iAA_i] = \sum_{v \in \mathcal G}    \P\{ v = A\} E \Big[ N_i^{(v)}(N_i^{(v)} -1) \Big| v= A\Big].
\ee
\end{proposition}

\begin{proof}[Proof of proposition \ref{prop:triples}]
We first assume that the types $A_i \neq A_j$, and we start computing from equation  \eqref{eq:condmass}:
\allowdisplaybreaks
\begin{align*}
[A_iAA_j] = \E(A_iAA_j)
&= \sum_{v \in \mathcal G}   \P\{ v = A\}  \sum_{ (v_1, v, v_3) \in \Pi_3} P \{ v_1 = A_{i},  v_3=A_{j} | v = A\} \text{ by \eqref{eq:condmass}}\notag \\
&= \sum_{v \in \mathcal G}  \P\{ v = A\}  E \Big[ \sum_{v_1: v_1\sim v} \sum_{v_3: v_3 \sim v} \mathbbm1\{ v_1 =A_i, v_3=A_j\} \Big| v= A\Big]\notag \\
&= \sum_{v \in \mathcal G}    \P\{ v = A\} E \Big[ N_i^{(v)}N_j^{(v)} \Big| v= A\Big].
\end{align*}
Similarly, we repeat the computation for when $A_i = A_j$. The change is during the second equality above, when we compute the double sum.
\begin{align*}
[A_iAA_i] = \E(A_iAA_i)
&= \sum_{v \in \mathcal G}   \P\{ v = A\}  \sum_{ (v_1, v, v_3) \in \Pi_3} P \{ v_1 = A_{i},  v_3=A_{i} | v = A\} \text{ by \eqref{eq:condmass}}\notag \\
&= \sum_{v \in \mathcal G}  \P\{ v = A\}  E \Big[ \sum_{v_1: v_1\sim v} \sum_{v_3: v_3 \sim v, v_3 \neq v_1} \mathbbm1\{ v_1 =A_i, v_3=A_i\} \Big| v= A\Big]\notag \\
&= \sum_{v \in \mathcal G}  \P\{ v = A\}  E \Big[ \sum_{v_1: v_1\sim v}  \mathbbm1\{ v_1 =A_i\} \sum_{v_3: v_3 \sim v, v_3 \neq v_1}\mathbbm1\{ v_3=A_i\} \Big| v= A\Big]\notag \\
&= \sum_{v \in \mathcal G}    \P\{ v = A\} E \Big[ N_i^{(v)}(N_i^{(v)} -1) \Big| v= A\Big].
\end{align*}
\end{proof}

In fact fixing an arbitrary network and disease parameters allows us to extract the expected values of triples that drive the epidemic (e.g. for an SIS epidemic these are $[SIS](t)$ and $[ISI](t)$) directly from sampling multiple realisations of the stochastic epidemic. Sampling across a discrete time set can then give the empirical conditional distribution of
 \[(N_1^{(v)} , N_2^{(v)}, \ldots , N_m^{(v)} ) | v = A \]
and also a way to estimate $\P\{v = A\}$ on these times, which can then be used 
in Proposition~\ref{prop:triples}. Networks with enough symmetries (i.e. if nodes are exchangeable) and a large number of nodes can give a good approximation of the empirical distribution with just one simulation.

These expected triple counts in time can then be fed into the pairwise model \eqref{eq:PW_SIS_S}-\eqref{eq:PW_SIS_II}  and hence its numerical integration is possible without a closure. In fact, this leads to a very accurate mean-field model as shown in Figure~\ref{fig:triplecount_into_pairwise} for three different network models. Of course, the applicability of this method is limited as it relies on output from the simulation. Nevertheless, it shows that being able to accurately approximate triples is key to derive accurate mean-field models.

\begin{figure}[h!]
\centering
\includegraphics[width=0.325\textwidth]{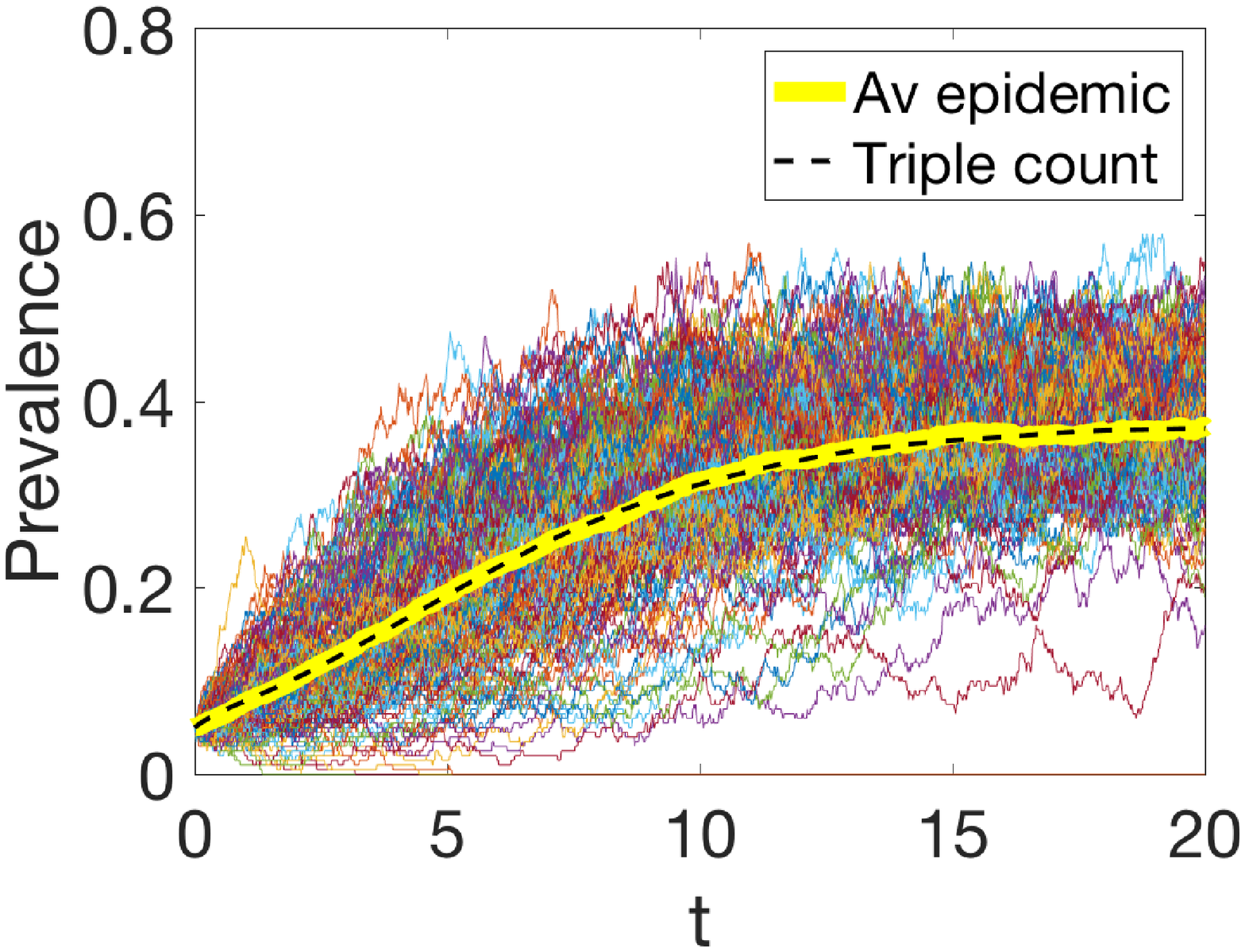}
\includegraphics[width=0.325\textwidth]{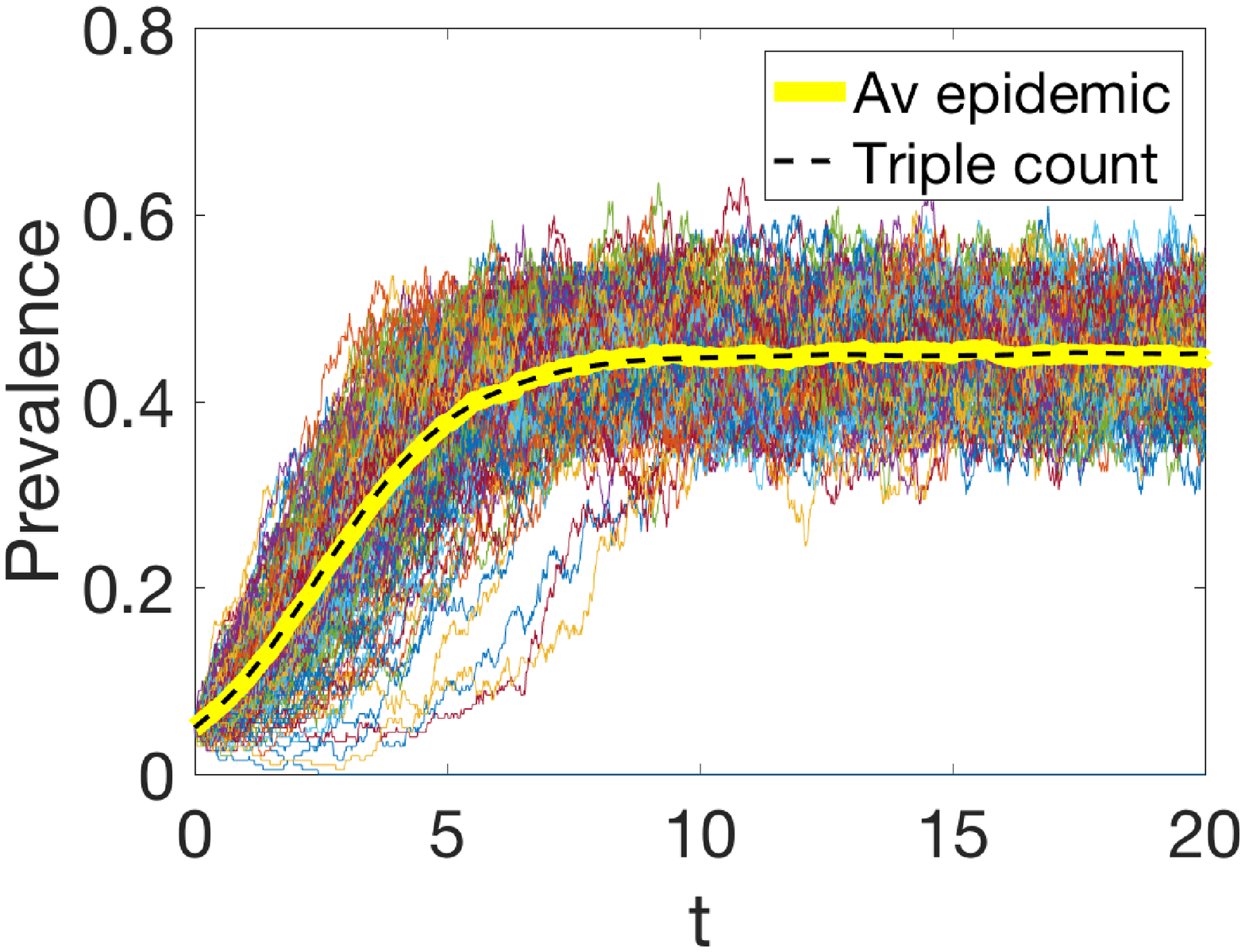}
\includegraphics[width=0.325\textwidth]{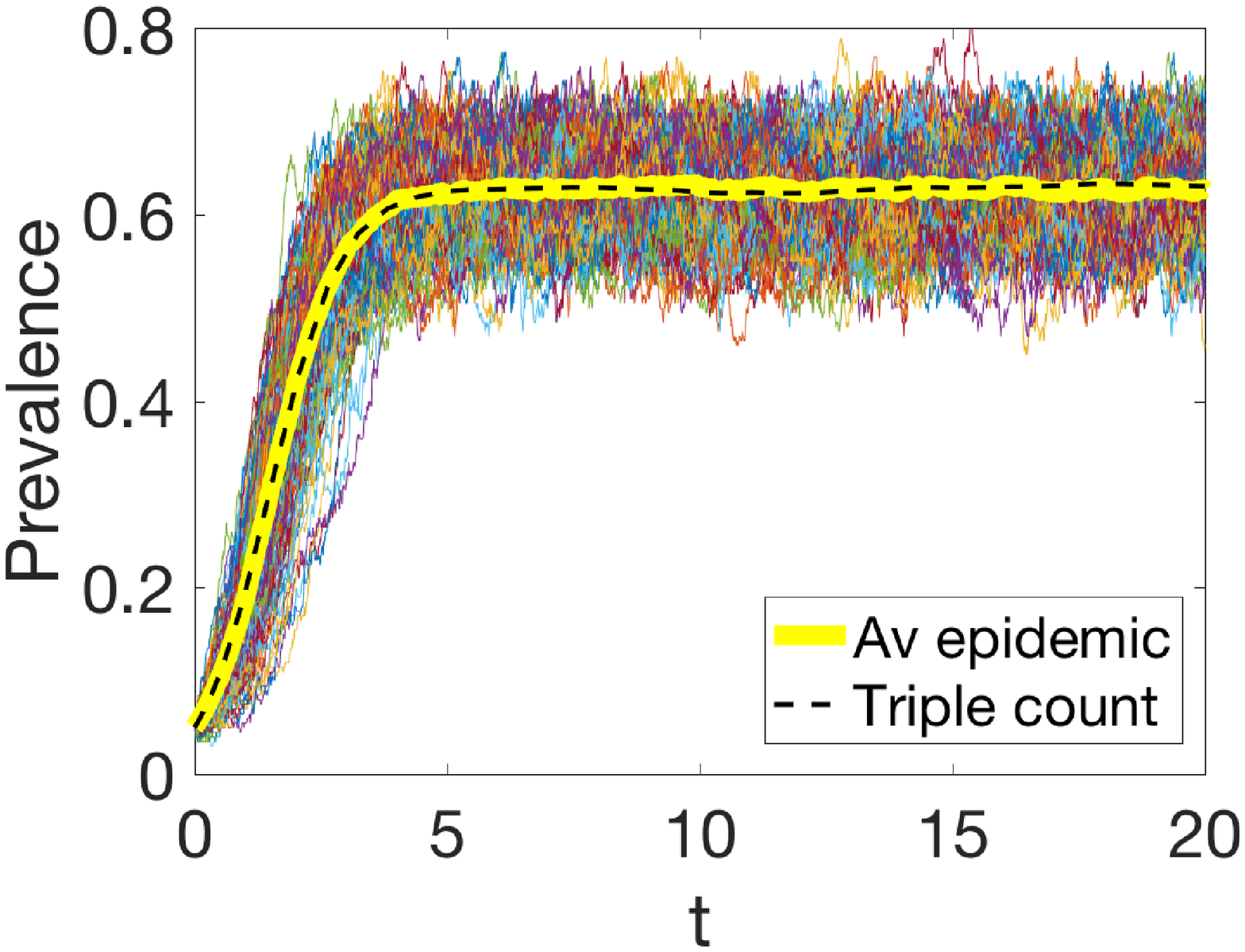}
   \centering
%   \begin{subfigure}[b]{0.3\textwidth}
%       \centering
%       \includegraphics[width=\textwidth]{tauvsC}
%   \end{subfigure}
%   ~
%   \begin{subfigure}[b]{0.3\textwidth}
%       \centering
%       \includegraphics[width=\textwidth]{tauvsCba}
%
%   \end{subfigure}
%   ~
%   \begin{subfigure}[b]{0.3\textwidth}
%       \centering
%       \includegraphics[width=\textwidth]{kvsp}
%   \end{subfigure}
    \caption{Individual realisations of the stochastic SIS model (multicoloured noisy lines) on regular (left), Erd\H{o}s-R\'enyi (middle) and bi-modal (right) networks. The average
    of all realisations is also shown (thick continuous line) along with the solution of the pairwise model \eqref{eq:PW_SIS_S}-\eqref{eq:PW_SIS_II} (thin dashed line) with the triples $[SSI]$ and $[ISI]$ taken as the expected value computed across all 200 individual epidemic simulations. In all figures $N=200$, $\tau=0.5$ and $\gamma=1$. Regular and Erd\H{o}s-R\'enyi networks have $\langle k\rangle=4$ while the bi-modal network has half of the nodes with degree 4 and the the other half with degree 8. Each epidemic starts with 10 infected nodes chosen uniformly at random and only epidemics that reach 20 infected nodes count towards the average with time reset to $t=0$ when the state with twenty infected nodes is first reached. }
    \label{fig:triplecount_into_pairwise}
\end{figure}

%\begin{figure}[ht]
%%%\epsfxsize=10cm   %width of figure - will enlarge/reduce the figures
%%%\epsfbox{fig3.eps}
%%%\figurebox{2cm}{3cm}{} %to have a box alone
%\centerline{\epsfxsize=4.1in\epsfbox{procs-fig1.eps}}
%\caption{First 3 normalized frequencies versus release location for
%clamped simply supported beam with internal slide
%release. \label{inter}}
%\end{figure}

\remark{Note that the formulas in Proposition \ref{prop:triples} have the sum over all vertices of $\mathcal{G}$. This is just the convenient way to write the formula for the sequel; the terms that need to be estimated
are then the $\P\{ v = A\}$ and $E \Big[ N_i^{(v)}N_j^{(v)} \Big| v= A\Big]$. We explain in the next sections how these quantities can be approximated.

The conditional expectation indicates that the important quantities clearly depend on the state of $v$ being $A$ and, per realisation, it is those sites that guide the value of the triple. Here is a different way that highlights this, assuming $A_i \neq A_j$, we get:
\begin{align*}
\#\{ A_iAA_j \}& = \sum_{v \in \mathcal G} \sum_{w: w \sim v}  \sum_{z: z \sim v} \mathbbm1\{v =A\}\mathbbm1\{w =A_i\}\mathbbm1\{z =A_j\}\\
&=  \sum_{v \in \mathcal G} \mathbbm1\{v =A\} N_i^{(v)}N_j^{(v)}=  \sum_{v: v =A}N_i^{(v)}N_j^{(v)}.
\end{align*}
This is an exact count, and the sum goes through the (random) set of nodes in state A. What we need for the mean-field model in this article is the expected value. Taking expectations cannot simplify the right-hand side any further as the sum is over a random index  and as such it does not commute with the expectation operation. This is the reason why we opt to use the formulas in Proposition \ref{prop:triples} rather than this more intuitive one.
}

%%%%%%%%%%%%%%%%%%%%%%%%%%%%%%%%%%%%%%%%%%%%%%%%%%%%%%%%%%%%%%%%%%%%%%%%%%%%%%%%
\section{Closures based on the distribution of states of the neighbours}\label{sec:closures}
%%%%%%%%%%%%%%%%%%%%%%%%%%%%%%%%%%%%%%%%%%%%%%%%%%%%%%%%%%%%%%%%%%%%%%%%%%%%%%%%
Proposition \ref{prop:triples} can be useful in certain situations, as can be seen in the following cases. In all that follow, a theoretical ansatz is made on the distribution of $(N_1^{(v)} , N_2^{(v)}, \ldots , N_m^{(v)} ) | v = A$ that can be used to further develop the formulas in Proposition \ref{prop:triples}. As a consequence, we can recover directly several closure formulas that are used in the literature, and showcase some further examples.

%%%%%%%%%%%%%%%%%%%%%%%%%%%%%%%%%%%%%%%%%%%%%%%%%%%%%%%%%%%%%%%%%%%%%%%%%%%%%%%%
\subsection{Closures for the multinomial distribution of states of the neighbours} \label{subsec:multi}
%%%%%%%%%%%%%%%%%%%%%%%%%%%%%%%%%%%%%%%%%%%%%%%%%%%%%%%%%%%%%%%%%%%%%%%%%%%%%%%%
Consider a a multi-type epidemic on a network, with possible states $A_1, \ldots, A_m$. For the {\bf Multinomial link model} we assume that for any given vertex $v$, the distribution of the $A_i$' s of its neighbouring vertices  depends  on the state of $v$, (say $v = A$) and conditional on the state of $v$, each neighbouring vertex takes a value $A_i$ from $\{ A_1, \ldots, A_m \}$ with probability $p^{(v, A)}_i$, independently of everything else.  The only condition required at this point is that
\be \label{eq:cond1}
\sum_{i=1}^m p^{(v, A)}_i =1.
\ee
This description is equivalent to the following 
\begin{assumption}\label{ass:m} Conditional on $v = A$, the conditional distribution
\be
(N_1^{(v)} , N_2^{(v)}, \ldots , N_m^{(v)} ) | v = A \sim {\rm Mult}(\deg v, m; p^{(v, A)}_1 , \ldots, p^{(v, A)}_m), \label{eq:multi}
\ee
i.e. the distribution of types around a node of type $A$ is multinomial with $m$ possible outcomes of $\text{deg}\, v$ independent experiments.
\end{assumption}

\begin{remark}
When $\deg v$ is constant (and not random), the marginal distributions of $N_i^{(v)}$ are binomial and they were considered before. They were used to derive
a triple closure for simple epidemics (see \cite{rand1999correlation}, chapter 4, p.104). In that notation \cite{rand1999correlation}, this distribution is denoted by $Q_v(i | j)$. Correlations were introduced to take into account the fact that the degree creates constraints, turning it into the multinomial model. The authors also discuss Poisson closures as viable approximations and we also do that at the end of the section, re-deriving their formulas.
\end{remark}

As we mentioned, Assumption \ref{ass:m} implies that the (conditional) marginal distributions of the coordinates are binomial and we can immediately
obtain that
\[
E(N_i^{(v)}| v = A) = p_{i}^{(v, A)}\deg v, \text{ for all } i \in \{ 1, \ldots, m \}.
\]
Moreover, a direct computation using \eqref{eq:multi} gives 
\[
E \Big[ N_i^{(v)}N_j^{(v)} \Big| v= A\Big] = p_{i}^{(v, A)}p_{j}^{(v,A)}\deg v (\deg v - 1).
\]
%The interested reader can find the calculation for these facts in the Appendix.

The conditional  distribution is multinomial, therefore the conditional expectation is the expectation of products of binomially distributed coordinates of the multinomial distribution \eqref{eq:multi}.
A different way to write $E[ N_i^{(v)}N_j^{(v)}| v= A] = \text{Cov}_{v = A}(N_i^{(v)}, N_j^{(v)}) + E[ N_i^{(v)}| v= A]\, E[N_j^{(v)} | v= A]$, which corresponds to the approach in  \cite{rand1999correlation}. %\cite{adv-exo-theo}.
Therefore when we substitute in \eqref{eq:more-ex} for the case where $A_i \neq A_j$ we get that
\allowdisplaybreaks
\begin{align}
[A_iAA_j] &=  \sum_{v \in \mathcal G}    \P\{ v = A\} p_{i}^{(v, A)}p_{j}^{(v,A)}\deg v (\deg v - 1). \label{eq:pre-approx}
\end{align}

%The final step is to substitute the symbols in line \eqref{eq:pre-approx} with some mean-field approximations, coming directly from the network of size $N$.
The final step is to derive as accurate as possible approximations for the probabilities/quantities in \eqref{eq:pre-approx}.

\begin{remark} Consider SIS epidemics on a regular network of degree $n$. There are only two types, so the distribution of types around (say) S nodes is completely characterised by the number (and therefore the distribution) of type I around S nodes.  If the multinomial model assumption is accurate, then that distribution of I's should be Binomial with $n$ trials. Indeed this can be numerically verified. Figure~\ref{fig:REG_Poisson_Binom_neighbours} shows histograms of the marginal distributions of the number of I nodes around susceptible nodes for an SIS epidemic simulated on a regular network. The probability of the binomial distribution is estimated using formula \eqref{eq:pi}, and it fits well to the simulated data. This particular fit suggests that Assumption \ref{ass:m} is not unrealistic in many situations.
\end{remark}

\begin{figure}[h]
\centering
\includegraphics[height=0.6\textwidth]{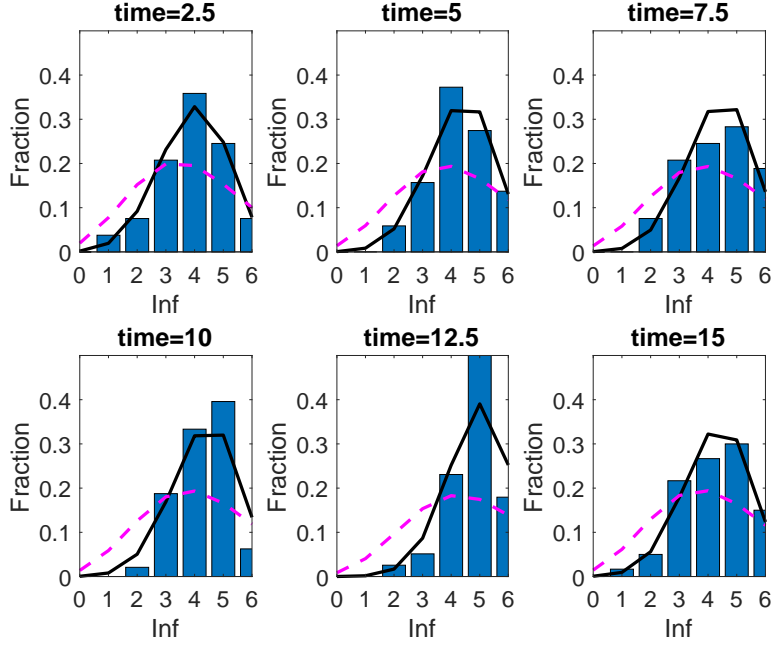}
\centering
\caption{The distribution of infected nodes around susceptible nodes at different points in the epidemic on a Regular network with $N=200$ nodes with $\langle k \rangle=6$, $\gamma=1$ and $\tau=0.75$. This is base on one single realisation of the epidemic on one single realisation of the network. The continuous  and dashed lines stand for the Binomial (with $\langle k \rangle=6$ trials and probability of success $[SI]/(\langle k \rangle [S])$ ) and Poisson distributions respectively (with rate $[SI]/[S]$), where the pair and single are taken from simulation. Notice how the binomial mass function is a better fit for this data than the Poisson one.}
\label{fig:REG_Poisson_Binom_neighbours}
\end{figure}

At this point we carefully explain the remaining assumptions that lead to the closures we are concerned with.

\begin{assumption} \label{ass:1}
\[
  \P\{ v = A\} \sim \frac{[A]}{N},\quad  { \rm for \,\,\, all \quad } v \in \mathcal G.
\]
\end{assumption}
We emphasise that this is not an assumption that applies to all networks, as it suggests that at any given time $t$, that any node has the same probability of being (say) infected; intuitively this fails when in networks with high degree heterogeneity, as well-connected or high-degree nodes have higher probability of being in a certain state. This is a reasonable assumption on networks with low degree heterogeneity.

%{\color{red} This should also be a good approximant in networks were the degree is near constant (say degree $k$ with prob. $1-\varepsilon$ and $k\pm 1$ with prob. $\varepsilon/2$) , or ergodic, like Erdos-Renyi graphs.}

The second assumption is about the values of $p_i^{(v,A)}$ in the conditional multinomial distribution.
\begin{assumption}\label{ass:2} 
We assume that
\[
p_i^{(v,A)} \equiv p_i^{(A)}.
\]
In other words,  the parameters of the multinomial distribution do not depend on $v$ in any way, but only on the state  $v$ is in.
\end{assumption}

Under Assumptions \ref{ass:m}, \ref{ass:1} and \ref{ass:2} we can first derive approximations for $p_{i}^{(A)}$:
\begin{align*}
[A_iA] &=\E(A_iA) = \E \Big( \sum_{v \in \mathcal G} \sum_{w: w\sim v} \mathbbm1\{ w=A_i, v = A \}\Big)\\
&=  \sum_{v \in \mathcal G} \P\{ v = A\} E(N_i^{(v)}| v = A)\approx \frac{[A]}{N} \sum_{v \in \mathcal G} \text{deg}\, v \, p_{i}^{(A)}\\
&= [A] \, p_{i}^{(A)} \langle\text{deg} \, v \rangle.
\end{align*}
Above we defined $ \langle\text{deg} \, v \rangle$ as the average degree of the network. The above computation implies
\be
\label{eq:pi} p_{i}^{(A)} \approx \frac{[A_iA]}{\langle \text{deg}\, v\rangle [A]}, \text{ for all } i \in\{ 1,2, \ldots, m\}.
\ee

We can now conclude computation \eqref{eq:pre-approx} under these approximations:
\begin{align*}
[A_iAA_j] &\approx \sum_{v \in \mathcal G} \frac{(\text{deg} v -1) \text{deg} v}{\langle \text{deg} v \rangle^2}\frac{1}{N}\frac{[A_iA][AA_j]}{[A]} \notag \\
&= \frac{[A_iA][AA_j]}{[A] \langle\text{deg} \, v \rangle^2 }\Big(\frac{1}{N}\sum_{v \in \mathcal G}(\text{deg} v -1) \text{deg} v\Big)\notag\\
&= \frac{[A_iA][AA_j]}{[A] \langle\text{deg} \, v \rangle^2 } \langle \deg v(\deg v -1) \rangle.
\end{align*}

Now in the case where $A_i = A_j$ we use \eqref{eq:more-ex2} and have
\begin{align*}
[A_iAA_i]&=\sum_{v \in \mathcal G}\P\{ v = A\} E \Big[(N_i^{(v)})^2 - N_i^{(v)}\Big| v= A\Big] \\
&=\sum_{v \in \mathcal G}\P\{ v = A\} \Big(\text{Var}(N_i^{(v)}| v=A) + (\E N_i^{(v)}|v=A)^2 - \E(N_i^{(v)}|v=A)\Big)\\
&= \sum_{v \in \mathcal G}\P\{ v = A\} ( p_i^{(A}(1-p_i^{(A)})\deg v +( p_i^{(A)} \deg v )^2 - p_i^{(A)} \deg v )\\
%&= \sum_{v \in \mathcal G}\P\{ v = A\}(-(p_i^{(A)})^2\deg v +( p_i^{(A)} \deg v )^2)\\
&= \sum_{v \in \mathcal G}\P\{ v = A\}(p_i^{(A)})^2\deg v( \deg v  -1).
\end{align*}
Substitute in the approximations for $\P\{ v = A\}$ and $p_i^{(A)}$ to obtain
$$
[A_iAA_i] \approx \frac{[A]}{[N]}  \frac{[A_iA]^2}{\langle \deg v\rangle^2 [A]^2}  \sum_{v \in \mathcal G} \deg v( \deg v  -1) = \frac{[A_iA]^2}{[A]\langle \deg v\rangle^2} \langle \deg v (\deg v-1) \rangle.
$$
Thus starting from Proposition \ref{prop:triples} and using the above assumptions we have proved the following.

\begin{theorem}\label{thm:closure}
Under Assumptions \ref{ass:m}, \ref{ass:1} and \ref{ass:2} the expected value of triples can be given as
\be
[A_iAA_j] = \frac{[A_iA][AA_j]}{[A] \langle\text{deg} \, v \rangle^2 } \langle \deg v(\deg v -1) \rangle \label{eq:gen-mul}
\ee
and
\be \label{eq:ii}
[A_iAA_i] = \frac{[A_iA]^2}{[A]\langle \deg v\rangle^2} \langle \deg v (\deg v-1) \rangle.
\ee
\end{theorem}

Applying this theorem to SIS dynamics we obtain

\begin{corollary}\label{cor:SISclosure}  For SIS epidemics, closures \eqref{eq:gen-mul}, \eqref{eq:ii} become
\be
[SSI] \approx \frac{[SS][SI]}{[S] \langle \deg \, v \rangle^2 } \langle \deg v(\deg v -1) \rangle,
\label{eq:SSI_closure_degv}
\ee
and
\be
[ISI]  \approx \frac{[SI]^2}{[S] \langle \deg \, v \rangle^2 } \langle \deg v(\deg v -1) \rangle.
\label{eq:ISI_closure_degv}
\ee
\end{corollary}

\begin{remark} The formulas in Corollary \ref{cor:SISclosure} are strikingly similar to the ones in article \cite{simon2015super}, which was obtained by decomposing triples based on the degree of the susceptible nodes. However, there is an a-priori  significant difference, as the degree moments here are over all nodes in the network, not just those in state S at time $t$. If the nodes in state S at any given time in the network is a statistically significant sample of nodes in the network and the nodes are exchangeable (e.g. low-degree heterogeneity), the moments here and the moments used in \cite{simon2015super} are very close to each other and the two closures will perform in a similar way.
\end{remark}

\begin{example}($n$-regular networks $\mathcal G$) If $\mathcal G$ is $n$-regular, i.e. all nodes have degree $n$, then $\deg v = n = \langle\deg v\rangle$. Assume for the moment that Assumption \ref{ass:2} does not necessarily hold. We may still proceed with the computation of $[A_iA]$, but now we would have
\[
[A_iA] \approx \frac{[A]}{N} \sum_{v \in \mathcal G} p_{i}^{(v, A)}\deg v,
\]
 and it can be pulled out of the sum, and after an elementary algebraic manipulation we have that
\[
\frac{1}{N} \sum_{v \in \mathcal G} p_{i}^{(v, A)} \approx \frac{[A_iA]}{n[A]}.
\]
Therefore the statistical average of $p_i^{(v, A)}$ (on the left hand side) is only a function of the state (right-hand side) which coincides with the expression in \eqref{eq:pi}.
%In this respect, Assumption \ref{ass:2} does not impose any unnatural conditions and .

For $\mathcal G$ $n$-regular, equation \eqref{eq:gen-mul} becomes the classical binomial closure 
\be
[A_iAA_j]\approx\frac{n-1}{n}\frac{[A_iA][AA_j]}{[A]},
\label{eq:closure_reg_n_takeaway_one} 
\ee
which can be found in \cite{rand1999correlation,keeling1999effects}.
\end{example}

Equation \eqref{eq:pi} is only an approximation. As it is stated, it violates the consistency condition \eqref{eq:cond1}. That can be easily accounted for by defining a normalised version. There are several ways to do this, each one leading to a different closure. We conclude this section with one such closure for SIS epidemics.
Since there are only two possible states one way to normalise probabilities is to assume that
\be \label{eq:onlyI}
p^{(S)}_I=\frac{[SI]}{[S]\langle \deg v\rangle} \text{ and } p^{(S)}_S = 1 - p^{(S)}_I.
\ee
With this, we have the following.
\begin{corollary}
Under approximation \eqref{eq:onlyI}
\be
[SSI] \approx [S] p^{(S)}_I(1-  p^{(S)}_I)  \langle \deg v(\deg v -1) \rangle,
\quad
%\be
[ISI] \approx [S](p^{(S)}_I)^2\langle \deg v(\deg v -1) \rangle.
\label{eq:closure_p_I}
\ee
\end{corollary}
Note that while we may also assume
\be \label{eq:onlyS}
p^{(S)}_S=\frac{[SS]}{[S]\langle \deg v\rangle} \text{ and } p^{(S)}_I = 1 - p^{(S)}_S,
\ee
this and the resulting closure are not driven by the important quantity in the epidemics, namely the number of $[SI]$ links since the quantity does not show up in the closure. As such, it is expected that a normalised closure using \eqref{eq:onlyS} would perform badly. Indeed this is supported by numerical evidence. Furthermore, we also tried to normalise by simply dividing each `probability' with the sum of the two but this led to an ill-behaved closed pairwise model and thus we omitted it here.

%%%%%%%%%%%%%%%%%%%%%%%%%%%%%%%%%%%%%%%%%%%%%%%%%%%%
\subsection{Closures for the Poisson distribution of states of the neighbours} \label{subsec:Poisson}
%%%%%%%%%%%%%%%%%%%%%%%%%%%%%%%%%%%%%%%%%%%%%%%%%%%%
In the regime where $n$ is large and the epidemic level is not high, we assume that conditional on $v = A$ the collection $\{ N_i^{(v)} | v = A \}_{1\le i \le m-1}$ is a collection of independent Poisson random variables; one type is the prevalent one which we omitted from the collection - for example in the initial stages of SIR epidemic, we can use this to model $N_R$ and $N_I$ since $S$ is prevalent. We assume that the parameters of the Poisson distribution do not depend on specific $v$ and we formalise this as follows.
\begin{assumption} \label{ass:poimixing}
\be
N_i^{(v)} | v = A  \sim { \rm Poisson}(\lambda_i^{(A)}).
\ee
\end{assumption}
For example if $\lambda_i^{(v, A)} < n^{\alpha} << n$ ($\alpha <1$), the independence assumption is not restrictive, since Poisson random variables are sharply concentrated around their mean, and the error of the approximation below can be controlled for types $1 \le i \neq j \le m-1$. 
In Figure~\ref{fig:ER_Poisson_neighbours} we show an example where the Poisson distribution does fit well  the marginal of the number of infected nodes around susceptible ones. The independence assumption is harder to verify numerically, but we present at the end of this section how the Poisson model naturally arises in Erd\H{o}s-R\'enyi graphs with low connection probability and prove the independence assumption. 

\begin{figure}[h]
\centering
\includegraphics[height=0.65\textwidth]{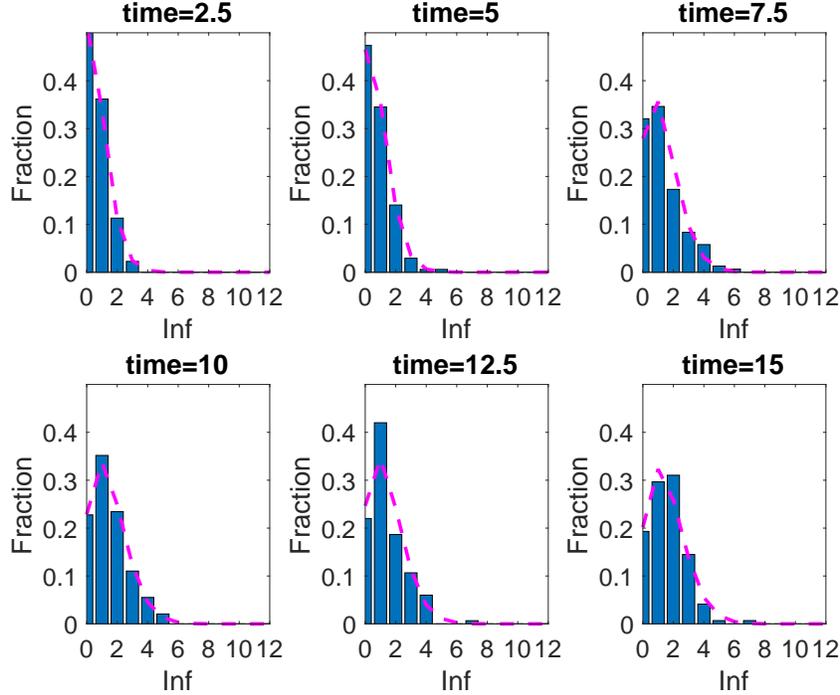}
\centering
\caption{The distribution of infected nodes around susceptible nodes at different points in the epidemic on a Erd\H{o}s-R\'enyi network with $N=200$ nodes with $\langle k \rangle=6$, $\gamma=1$ and $\tau=0.25$. Maximum degree in the network is twelve. This is based on one single realisation of the epidemic on one single realisation of the network. The dashed line  line stands for the Poisson distribution with mean given by $[SI]/[S]$  taken from simulation.}
\label{fig:ER_Poisson_neighbours}
\end{figure}

Assumption \ref{ass:poimixing} implies that
\[
E(N_i^{(v)}| v = A) = \lambda_i^{(A)}, \text{ for all } i \in \{ 1, \ldots, m \}.
\]
Moreover, the independence assumption on the coordinates implies the equality below,
\[
E \Big[ N_i^{(v)}N_j^{(v)} \Big| v= A\Big] = E \Big[ N_i^{(v)} \Big| v= A\Big]E\Big[N_j^{(v)} \Big| v= A\Big]= \lambda_i^{(A)}\lambda_j^{(A)}.
\]
Under Assumption \ref{ass:poimixing} we can first derive an approximation for $\lambda_i^{(A)}$:
\begin{align*}
[A_iA] &=\E(A_iA) = \E \Big( \sum_{v \in \mathcal G} \sum_{w: w\sim v} \mathbbm1\{ w=A_i, v = A \}\Big)\\
&=  \sum_{v \in \mathcal G} \P\{ v = A\} E(N_i^{(v)}| v = A) =  \sum_{v \in \mathcal G} \P\{ v = A\} \lambda_i^{(A)} = [A] \, \lambda_i^{(A)}.
\end{align*}
Thus we have
\be
\label{eq:lambdai} \lambda_i^{(A)} \approx \frac{[A_iA]}{[A]}, \text{ for all } i \in\{ 1,2, \ldots, m\}.
\ee
Starting again from \eqref{eq:more-ex} we obtain,
\[
[A_iAA_j]= \sum_{v \in \mathcal G} \P\{ v = A\} E \Big[ N_i^{(v)}N_j^{(v)} \Big| v= A\Big]
=\sum_{v \in \mathcal G} \P\{ v = A\}\lambda_i^{( A)}\lambda_j^{( A)} \approx [A]\lambda_i^{( A)}\lambda_j^{( A)}.
\]
We note that this relation holds also for $i=j$. Combining this equation with \eqref{eq:lambdai} leads to the following theorem.

\begin{theorem}\label{thm:Poissonclosure}
Under Assumption \ref{ass:poimixing} the expected value of triples can be given as
\be
[A_iAA_j] = \frac{[A_iA][A_jA]}{[A]}, \label{eq:gen-Poisson}
\ee
\end{theorem}
which can also be found in~\cite{rand1999correlation}.

Focus for the moment on Poisson closures in SIS epidemics.

\begin{corollary}\label{cor:SISPoissonclosure}  For SIS epidemics, closure \eqref{eq:gen-Poisson} yields
\be
[SSI] \approx \frac{[SS][SI]}{[S]}, \qquad [ISI]  \approx \frac{[SI]^2}{[S]}.
\label{eq:closure_Poisson}
\ee
\end{corollary}

The Poisson model is not as ad-hoc as it may initially seem, but it actually arises from the multinomial model in the case of Erd\H{o}s-Re\'nyi graphs with low average degree. For the remainder of this section we restrict the discussion to Erd\H{o}s-R\'enyi graphs of $N+1$ nodes, but with low average degree $\langle k \rangle = D$. 

In this case, we assume that the probability of a link being present in the network is $\frac{D}{N}$. The degree of each node $v$ 
is a random Binomial number $\deg v \sim \text{Bin}(N, \frac{D}{N})$. The law of rare events suggests that for $N$ large, 
the degree distribution is approximately Poisson$(D)$, so for this particular discussion we assume the network has marginal degree distribution
\[
\deg v \sim \text{Poisson}(D).
\]
For clarity of the exposition, focus only on $SIS$ epidemics on the network. We still operate under the multinomial link distribution (Assumption \ref{ass:m}), that is, 
\[
(N_S^{(v)}, N_I^{(v)}) | v = S \sim \text{Mult}(\deg v, 2; p_S^{(S)}, p_I^{(S)}). 
\]
However, here the number of trials of the multinomial is a random Poisson number $\deg v$. Given that the degree is $x \in \N$, then the distribution of types is multinomial. We first use the law of total probability to find the marginal distribution of $N_I^{(v)}$:
\begin{align*}
	\P\{ N_{I}^{(v)}  = k\} 
	&= \sum_{x = 0}^{\infty} P\{ N_{I}^{(v)}  = k | \deg v = x\} \P\{ \deg v = x \} \\
	&= \sum_{x = k}^{\infty} { x \choose k }( p_I^{(S)} )^k (1- p_I^{(S)})^{x - k} e^{-D} \frac{D^x}{x!}\\
	&= e^{-D}\frac{D^k  ( p_I^{(S)} )^k }{k!}  \sum_{x = k}^{\infty}  \frac{ (D(1- p_I^{(S)}))^{x - k}}{(x-k)!}=e^{- p_I^{(S)} D}\frac{D^k  ( p_I^{(S)} )^k }{k!},
\end{align*}
and therefore the conditional marginal $[N_{I}^{(v)} | v =S ] \sim \text{Poisson}(Dp_I^{(S)})$. 
Similarly, the marginal distribution of $N_{S}^{(v)}| v = S$ is that of a Poisson$(Dp_S^{(S)}))$ and conveniently, we may write 
\[
(N_{S}^{(v)}| v = S) =( \deg v - N_{I}^{(v)}| v = S) = \deg v - \text{Bin}(\deg v, p_I^{(S)}) \sim \text{Poisson$(Dp_S^{(S)}))$ }.
\]
Moreover, these two marginals are in fact independent. Here is the short calculation. Since $\deg v$ can take values in $\N$, as is Poisson distributed, the same holds for both $N_{S}^{(v)}$ and $N_I^{(v)}$. Then
\begin{align*}
\P\{ N_{S}^{(v)}= \ell, N_{I}^{(v)}=m\}& = \P\{ \deg v - N_{I}^{(v)} = \ell, N_{I}^{(v)}=m\} =  \P\{ \deg v= m + \ell, N_{I}^{(v)}=m\}\\
&= P\{ N_{I}^{(v)}=m |  \deg v= m + \ell\} \P\{  \deg v= m + \ell \}\\
&={{\ell + m} \choose m} (p_I^{(S)})^m  (1- p_I^{(S)})^\ell e^{-D}\frac{D^{\ell + m}}{(\ell + m)!}\\
&= e^{- p_I^{(S)}D} \frac{(p_I^{(S)} D)^m}{m!} \cdot  e^{-(1- p_I^{(S)})D} \frac{ (D(1- p_I^{(S)}))^\ell}{\ell!}\\
&= \P\{N_{I}^{(v)}=m\}\P\{ N_{S}^{(v)}= \ell\}.
\end{align*}

To summarise, on Erd\H{o}s-R\'enyi graphs with low average degree but many nodes, the fact that the marginal degree is Poisson distributed reduces the multinomial link model to the Poisson model. For the Erd\H{o}s-R\'enyi case, the assumption that the marginal distributions of types are Poisson, with rates given by $ \lambda_S^{(S)} = \langle k\rangle p_S^{(S)}$ and $ \lambda_I^{(S)} =  \langle k\rangle p_I^{(S)}$ respectively, and the fact that they are independent, is proven. 

Because of this, we expect that when the Poisson closure is used for the mean-field epidemic model on Erd\H{o}s-R\'enyi graphs, the approximation performs at least as well as the one using the moments closure given by Corollary \ref{cor:SISclosure}. This is also suggested by the simulations in the middle row of Figure~\ref{fig:incidence_pairwise_v_sim}, where the two curves lie on top of each other, and both approximate the semi-stable regime better than the binomial closure, when the infection rate is not too small. 

\section{Comparison of the closed pairwise model to stochastic simulations and further results}\label{sec:numerics_and_threshold}
%%%%%%%%%%%%%%%%%%%%%%%%%%%%%%%%%%%%%%%%%%%%%%%%%%%%
We performed a fair number of numerical test where we compared the accuracy of four closures (i.e. binomial, Poisson, moments and $p_I$ given by equations \eqref{eq:closure_reg_n_takeaway_one}, \eqref{eq:closure_Poisson}, \eqref{eq:SSI_closure_degv}-\eqref{eq:ISI_closure_degv} and \eqref{eq:closure_p_I}, respectively) by comparing the expected proportion of infected nodes in time from the closed system with that resulting from explicit stochastic simulation of the SIS dynamics on regular, Erd\H{o}s-R\'enyi and bi-modal networks. It is well known that agreement between closed pairwise systems and simulation tend to improve for large epidemics. In this case this is equivalent to either having dense networks or high rates of transmission, when assuming that the recovery rate is fixed. However, from the viewpoint of deriving closures, it is well-known that the most
interesting behaviour may be for networks with low average degree. Due to this our tests are conducted for networks with an average degree equal to six. The recovery rate is also fixed and the value of the transmission rate $\tau$ is varied to go from small to large epidemics. The results are shown in Figure~\ref{fig:incidence_pairwise_v_sim}.

Several observations can be made. As expected, for large epidemics the various closures lead to very similar output and the agreement with simulation is good. However, the Poisson closure leads to an over estimation in all cases, while the binomial closure does better especially for regular and bi-modal networks. For regular networks, the moments and the binomial closure coincide. In general, the binomial closure does better on regular networks. The new moments closure in many cases leads to results which are close to those based on the Poisson closure. One such case is when the  graph is Erd\H{o}s-R\'enyi with low average degree where the multinomial link model becomes the Poisson model, as we discussed in the previous section. 

In arbitrary networks of low average degree the independence assumption for the Poisson marginals is unrealistic. This is because the rates add up to the average degree but the probability of the sum of two independent Poisson to be equal to the degree of a node, for all nodes, with small fluctuations is small. In other words, it is expected that the binomial or moments approximations would perform better as the total degree of a node is built into the modelling. This can be seen from the closure's performance in Figure ~\ref{fig:incidence_pairwise_v_sim} on the regular and bi-modal degree networks when the average degree is low (first and last row of Figure~\ref{fig:incidence_pairwise_v_sim}).

When the infection rate $\tau$ of the SIS epidemics is low, the probability $p_I^{(S)}$ is low as well. Then the mass function of the binomial distribution with probability of success $p_I^{(S)}$ and $\langle k \rangle$ trials and the mass function of the Poisson distribution with rate $\langle k \rangle p_I^{(S)}$ are very close. This implies that the information needed for judging which closures are suitable for the epidemic model is not contained only in the marginal distributions of types around nodes.  

In Figure~\ref{fig:incidence_pairwise_v_sim} the first column shows the estimated mean field epidemic model using all closures for low $\tau$, and the moment closure (cyan line) does estimate better than, or as good as, the Poisson one, while the empirical distribution of infected nodes is well-approximated by both Binomial and Poisson. However, when the value of $\tau$ becomes high, the marginal distributions do contain more information. For example in Figure~\ref{fig:REG_Poisson_Binom_neighbours}, we see that the binomial mass function (black line) approximates the empirical distribution better than the Poisson (dashed magenta line) and this is reflected in the better approximation that can be seen in the first row of Figure~\ref{fig:incidence_pairwise_v_sim}. 

\begin{figure}[h!]
\centering
\includegraphics[width=0.325\textwidth]{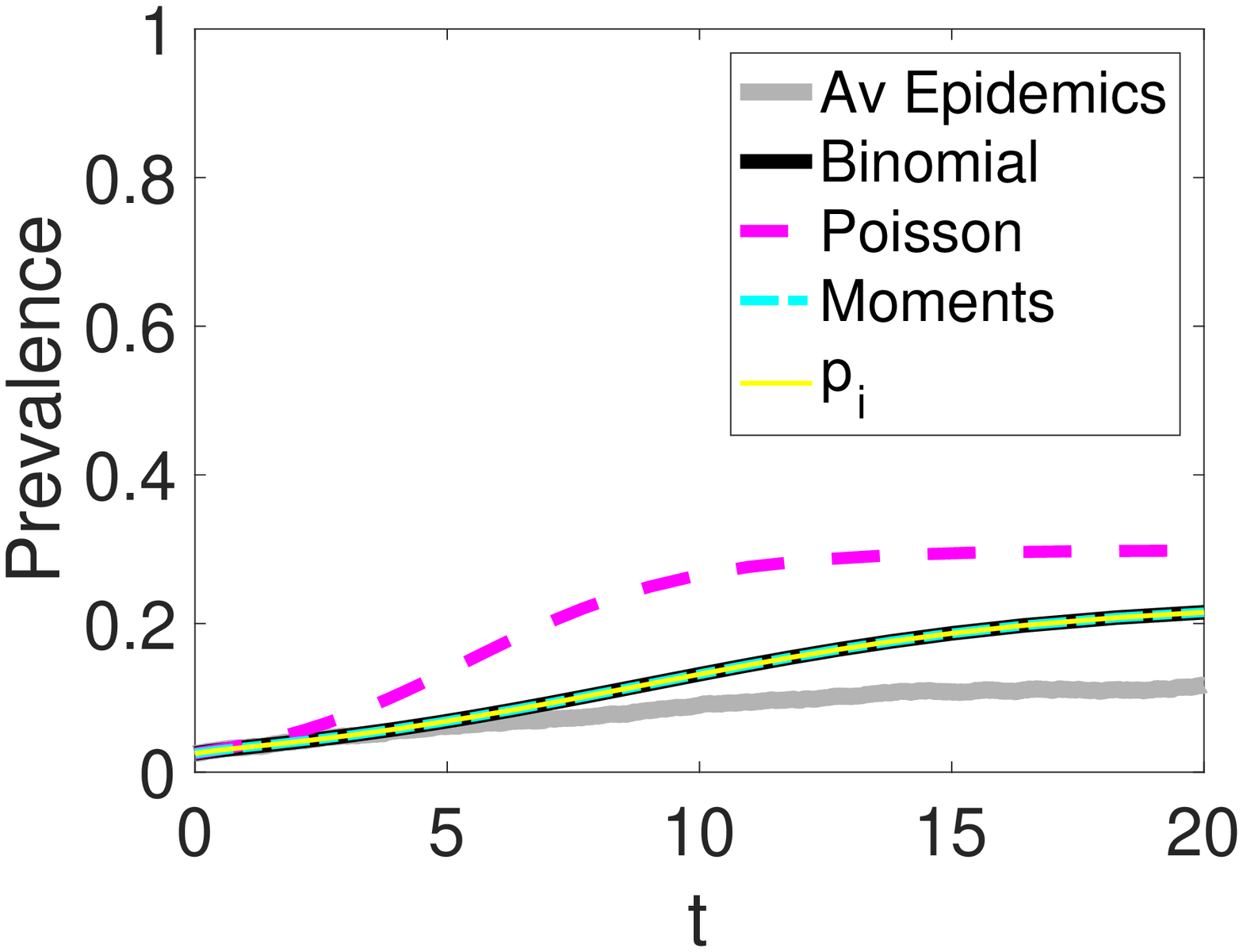}
\includegraphics[width=0.325\textwidth]{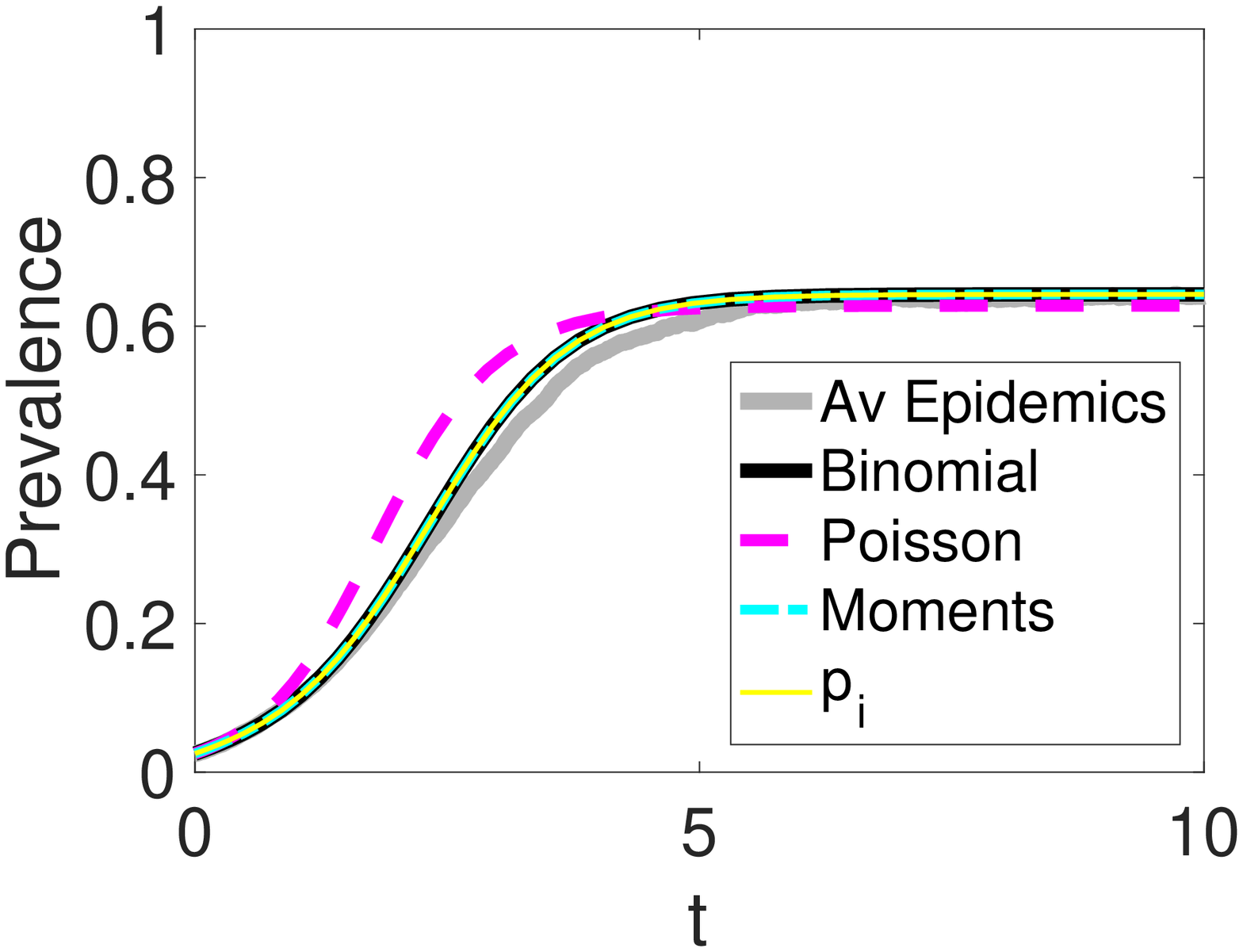}
\includegraphics[width=0.325\textwidth]{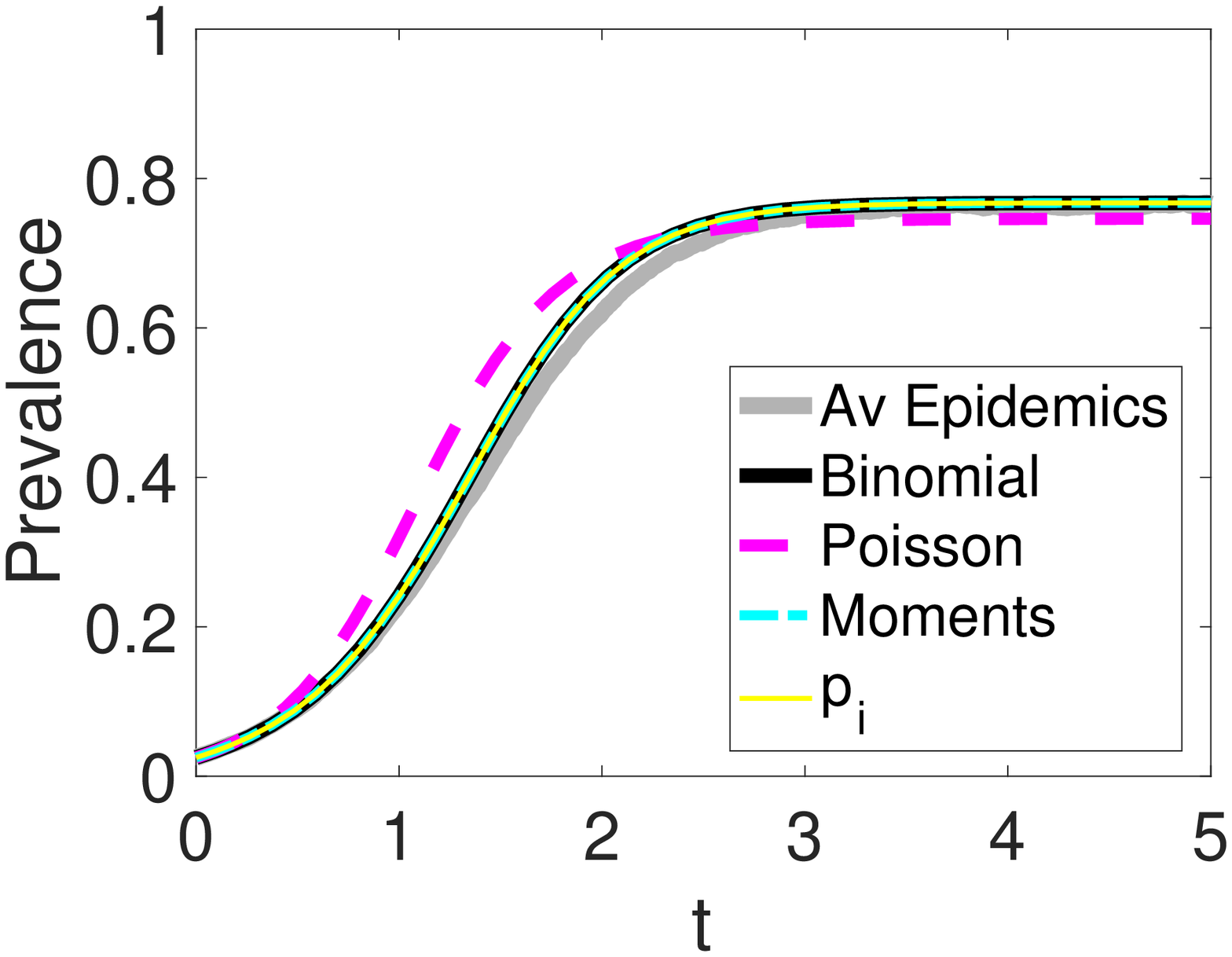}
\includegraphics[width=0.325\textwidth]{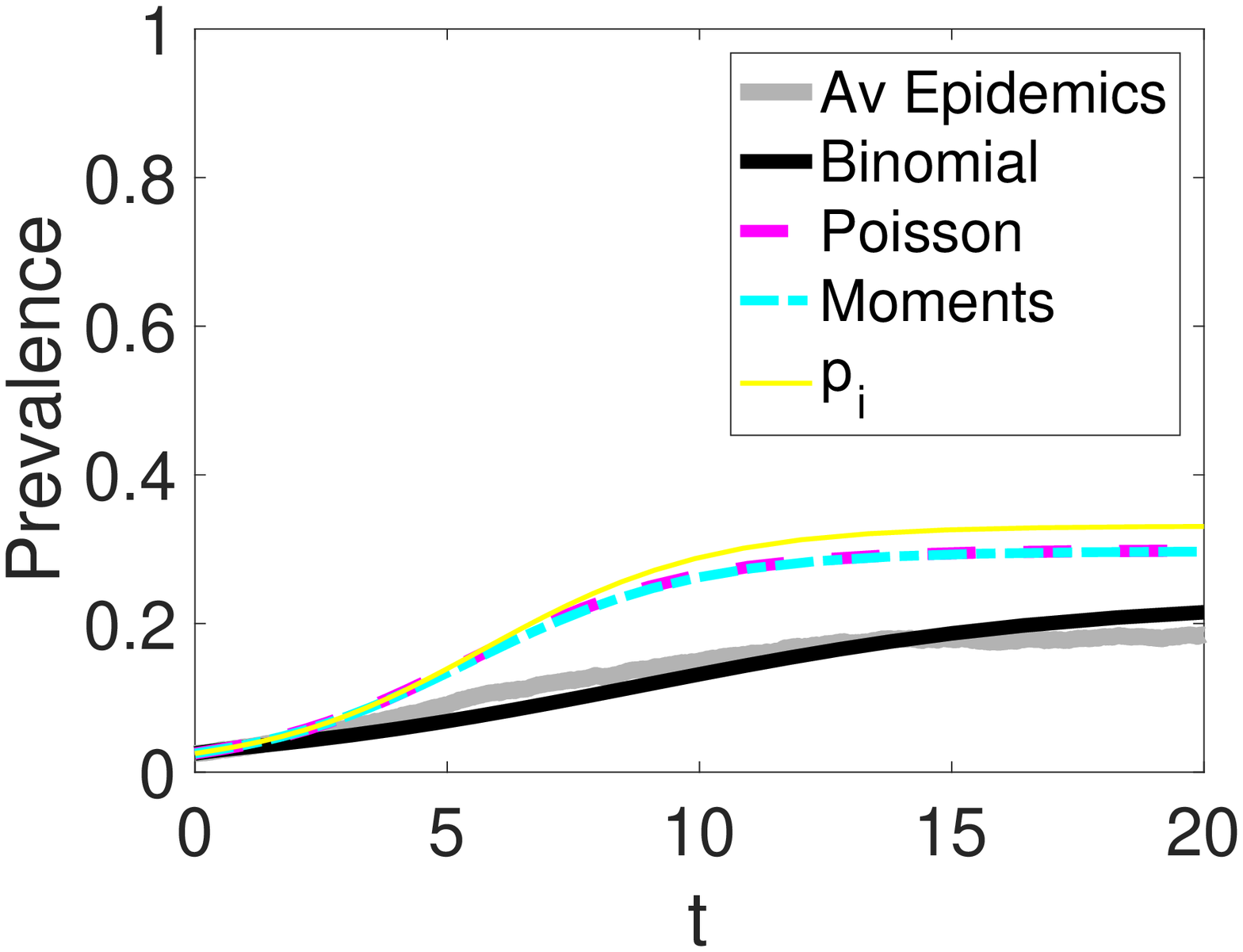}
\includegraphics[width=0.325\textwidth]{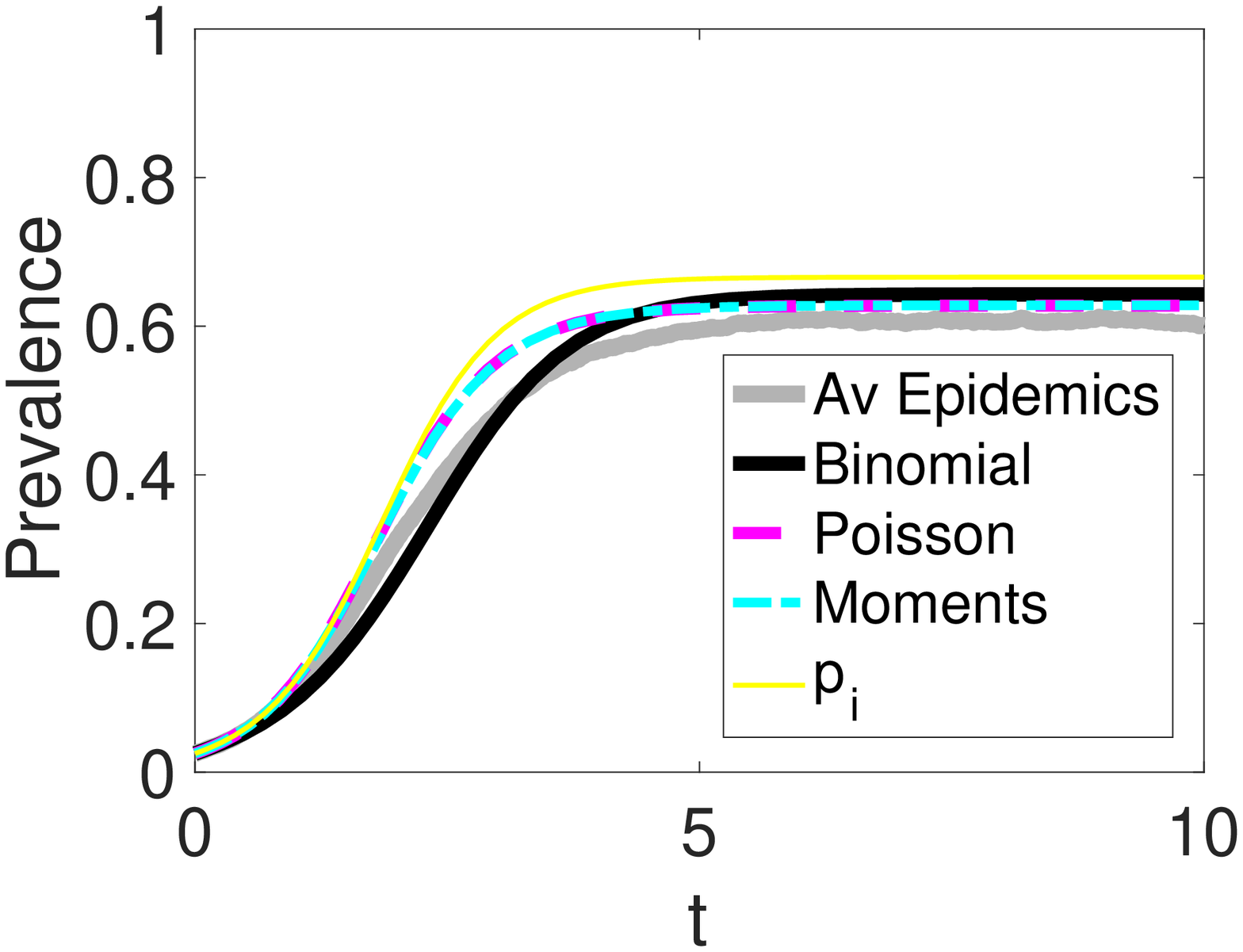}
\includegraphics[width=0.325\textwidth]{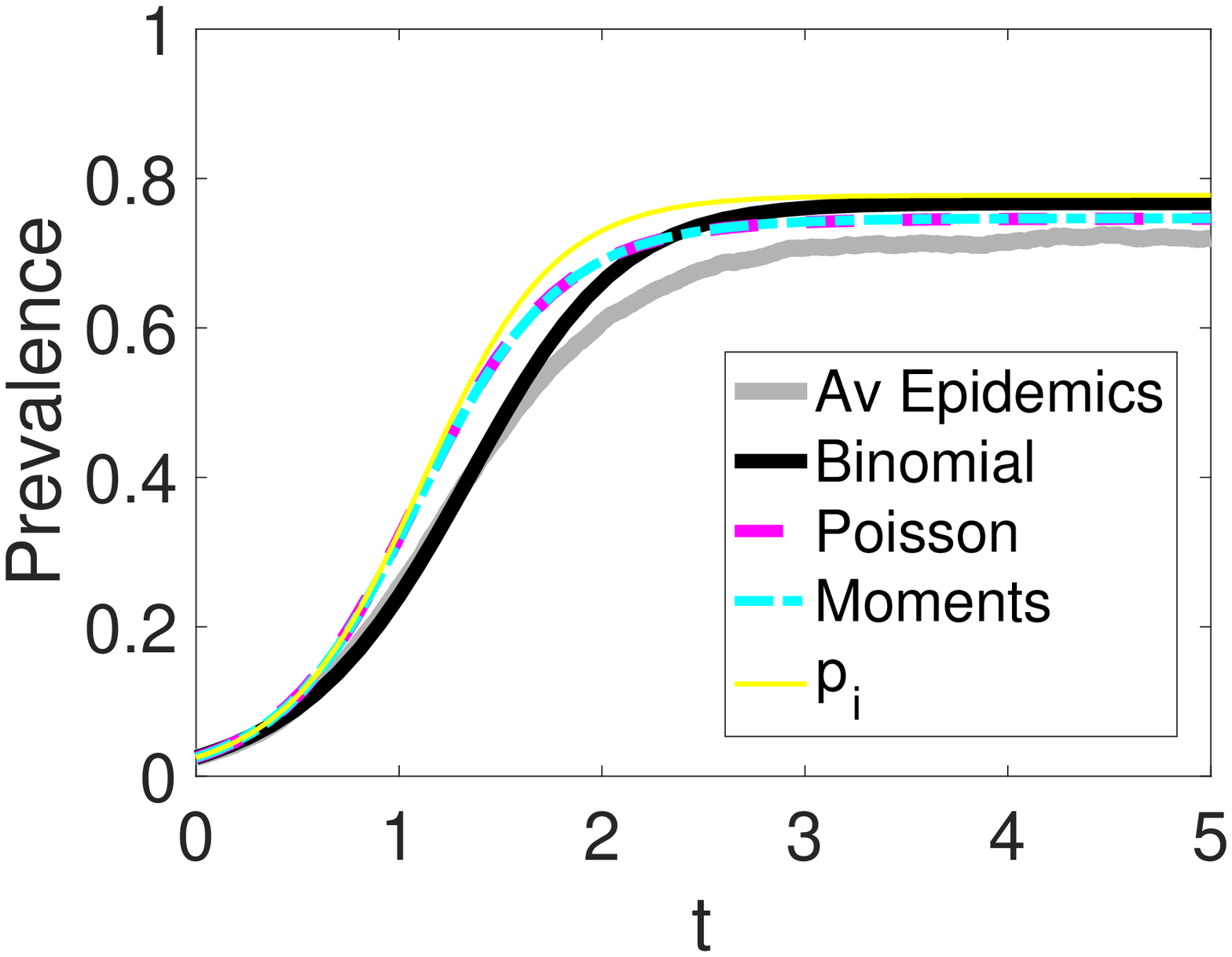}
\includegraphics[width=0.325\textwidth]{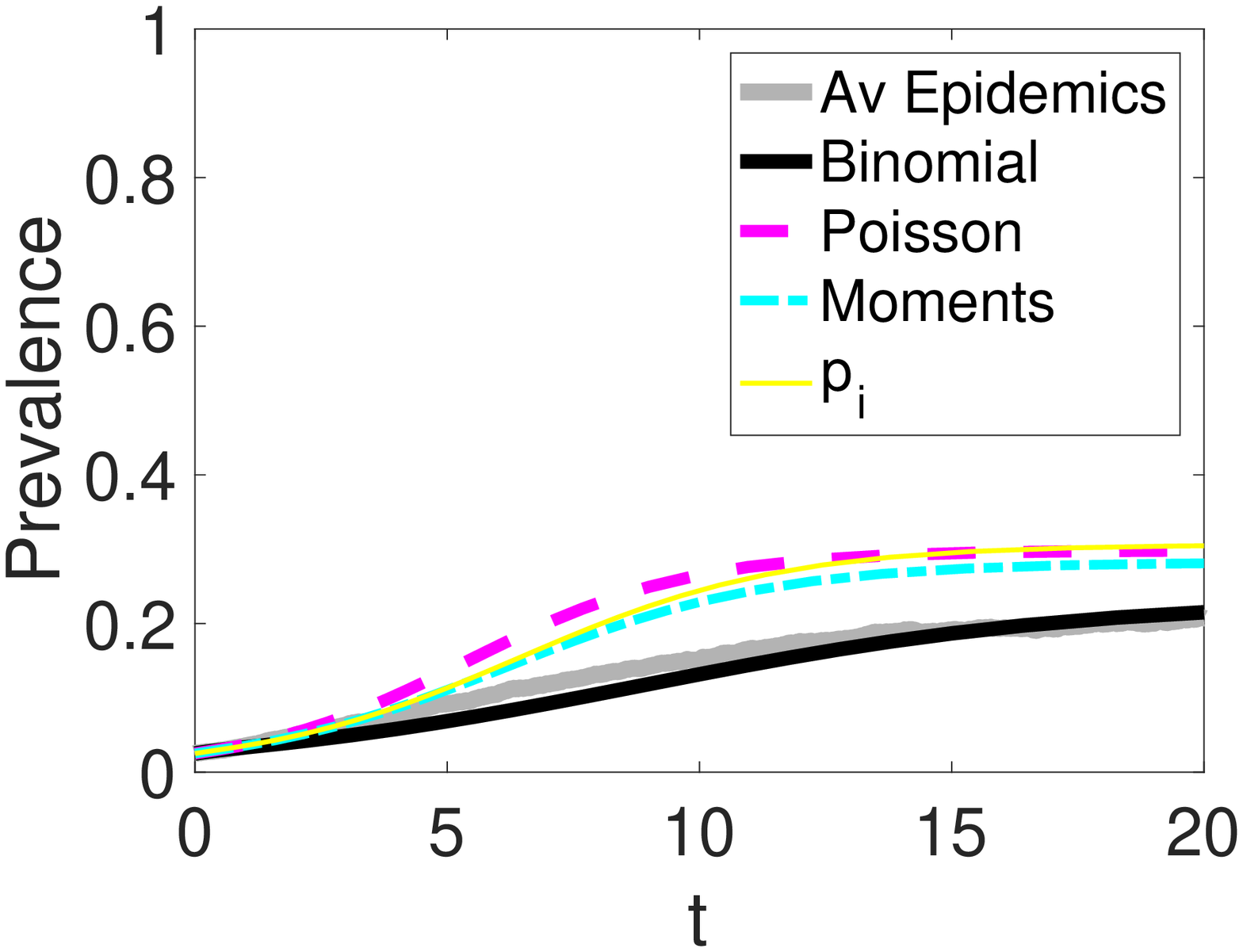}
\includegraphics[width=0.325\textwidth]{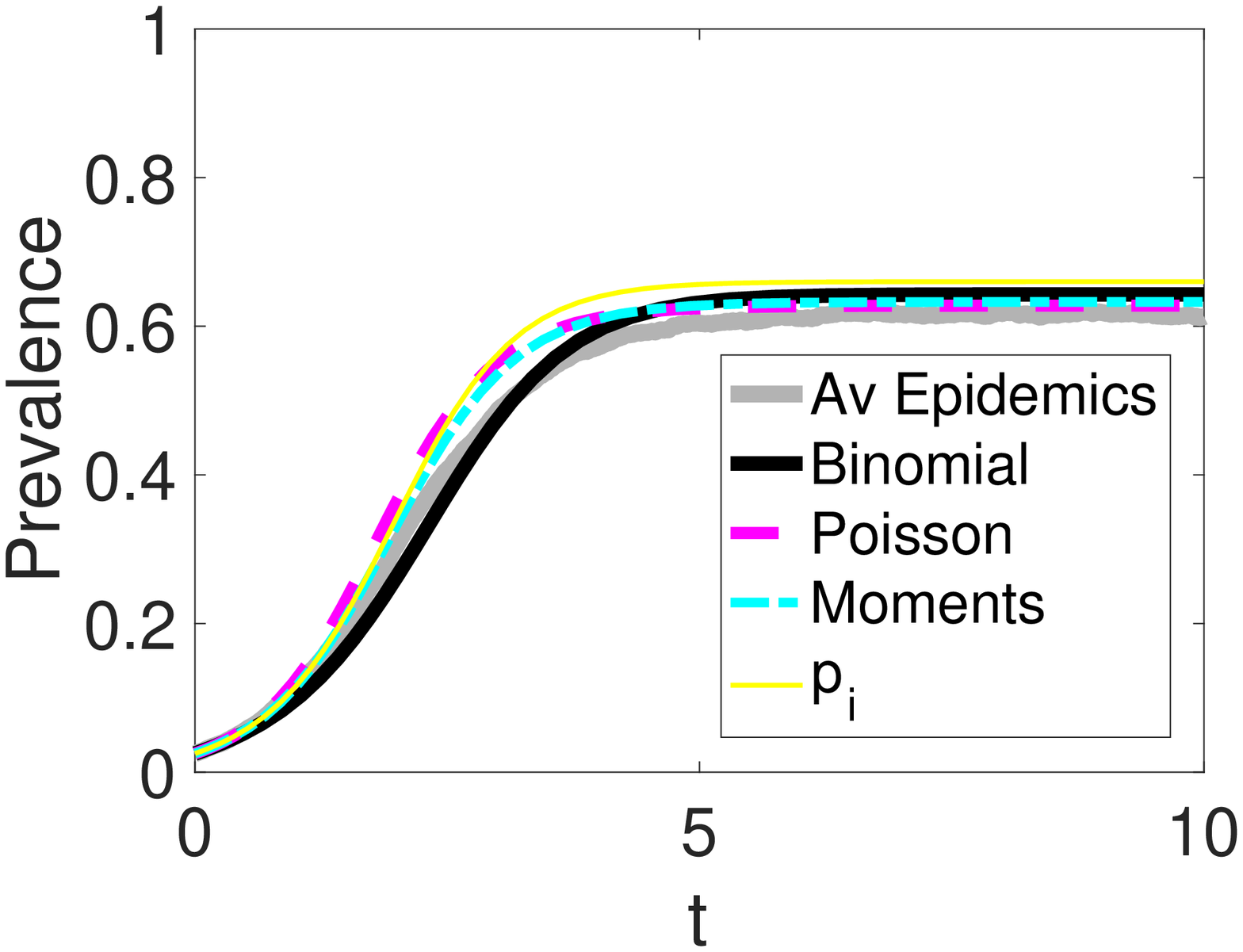}
\includegraphics[width=0.325\textwidth]{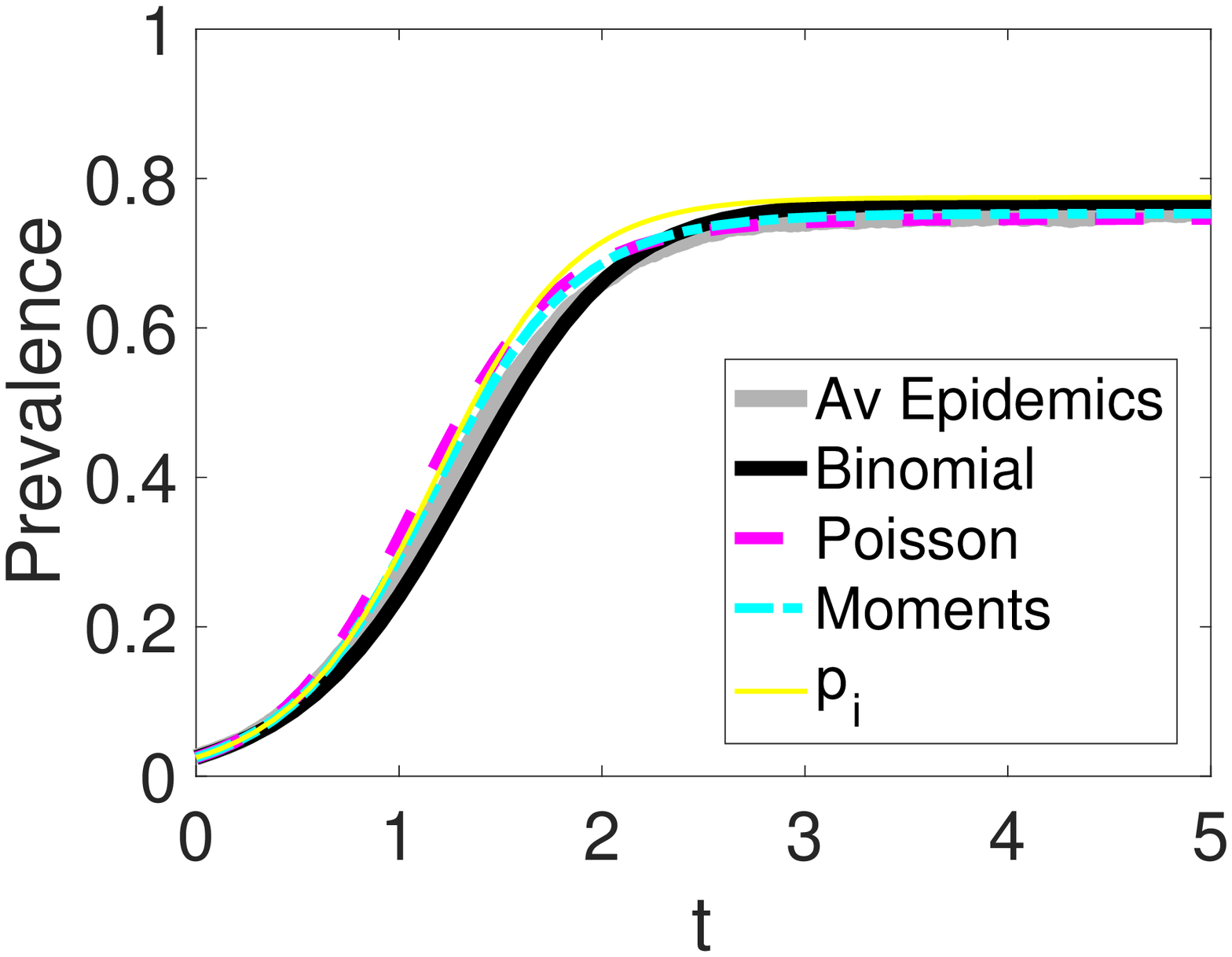}
   \centering
    \caption{The average of 200 individual realisations (thick grey line) together with output from four mean-field models (binomial (black), Poisson (dashed magenta), moments (cyan) and $p_I$(yellow) given by equations \eqref{eq:closure_reg_n_takeaway_one}, \eqref{eq:closure_Poisson}, \eqref{eq:SSI_closure_degv}-\eqref{eq:ISI_closure_degv} and \eqref{eq:closure_p_I}, respectively).
    The regular, Erd\H{o}s-R\'enyi and bi-modal networks are shown on the top, middle and bottom row, respectively. All networks have an average degree of $\langle k \rangle=6$ with the bi-modal having half of the nodes with degree four and the the other half with degree eight. In all figures $N=200$ and $\gamma=1$. The rate of infection in the left, middle and right columns is $\tau=0.25$, $\tau=0.5$ and $\tau=0.75$, respectively. Each epidemic starts with five infected nodes chosen uniformly at random and only epidemics that reach ten infected nodes count towards the average with time reset to $t=0$ when the state with ten infected nodes is first reached.}
    \label{fig:incidence_pairwise_v_sim}
\end{figure}

%\textbf{NICOS-XXX, please update this, makes sure it references correct figures and it says what you wanted} However, when the graph has a random degree sequence given by a Poisson distribution, then the Poisson approximation does well even in low degrees (\textcolor{red}{or even better?}). This is because the Poisson degree of a node can be viewed as a sum of two independent Poisson variables (in this case the one that represents the number of $S$'s and the one for the $I$'s). A network with an approximate Poisson degree distribution can be constructed the following way: Consider $N+1$ nodes (assumed large) and construct an Erd\H os-Renyi random graph on them, with probability of success  $\alpha/N$ to create a link.
%Then by the law of rare events, the degree of each node is approximately Poisson($\alpha$). 

%In Figure XXX you can compare the approximations of the various closures against the statistical mean of the epidemics and observe that for the Erd\H{o}s-R\'enyi
%network the Poisson is a better approximation than multinomial.

Applying the slightly more general closures \eqref{eq:SSI_closure_degv}-\eqref{eq:ISI_closure_degv} in the unclosed pairwise system \eqref{eq:PW_SIS_S}-\eqref{eq:PW_SIS_II} leads to following equation for the expected number of $[SI]$ pairs,
\be
\dot{[SI]}=\tau \mathcal{D} \frac{[SS][SI]}{[S]}- \tau \mathcal{D} \frac{[SI][SI]}{[S]}-\tau [SI]-\gamma [SI]+\gamma [II], \label{eq:SI_epi_threshold_1}
\ee
where
\be
\mathcal{D}=\frac{\langle \deg v(\deg v -1) \rangle}{\langle \deg v \rangle^2}.
\ee
Since the epidemic is driven by the number and growth rate of $[SI]$, we can analyse the equation above at the disease free steady state, i.e. $([S],[I],[SS],[SI],[II])=(N,0,\langle \deg v \rangle N,0,0)$.
This allows us to work out the rate of growth by looking at all the positive terms in~\eqref{eq:SI_epi_threshold_1} when evaluated at time $t=0$. This leads to
\be
\dot{[SI]} \sim \tau \mathcal{D} \frac{\langle \deg v \rangle N}{N}[SI]=\frac{\tau \langle \deg v(\deg v -1) \rangle}{\langle \deg v \rangle}[SI]. \label{eq:SI_epi_threshold_2}
\ee
However, the lifetime of an S-I link is simply $1/(\tau+\gamma)$ and hence, the average number of S-I links produced by a typical S-I link during its lifetime is
\be
\mathcal{R}=\frac{\tau}{\tau+\gamma}\frac{\langle \deg v(\deg v -1) \rangle}{\langle \deg v \rangle},
\ee
which is a well know quantity in the epidemics on networks mathematical theory~\cite{diekmann2000mathematical}. This really reassuring 
as it suggests that our newly derived closure is based on sound assumptions. 

%%%%%%%%%%%%%%%%%%%%%%%%%%%%%%%%%%%%%%%%%%%%%%%%%%%%
\section{Further extensions of the method} \label{sec:further_extensions}
%%%%%%%%%%%%%%%%%%%%%%%%%%%%%%%%%%%%%%%%%%%%%%%%%%%%
%
%%%%%%%%%%%%%%%%%%%%%%%%%%%%%%%%%%%%%%%%%%%%%%%%%%%%
\subsection{Conditioning on a link}
%%%%%%%%%%%%%%%%%%%%%%%%%%%%%%%%%%%%%%%%%%%%%%%%%%%%
%In the previous section we used the very specific assumption that the states of neighbours of a given site are drawn from a multinomial distribution, and this allows for the specific computation. By changing the conditional distribution in \eqref{eq:multi} to a model-driven one, one can simply continue the calculation from the first equality in \eqref{eq:pre-approx} in order to compute triple closures.

There are many ways to count triples in a network. In the multinomial link example we counted triples $[ABC]$ by focusing on the state of the middle node. A different way would be to count them by looking at the neighbours of node $v_2$, given that the link $(v_2, v_3) = (BC)$. In order to count this way we will use equation \eqref{eq:linky} by conditioning on the values of $v_2$ and $v_3$.

Let $\mathcal G$ be an $n$-regular graph.  Assume we want to compute closures of the form
$[A_i A B]$, so we are given that the state of $v_2 = A$ and state of $v_3 = B$,  where $A, B \in \{ A_1, \ldots, A_m\}$. Homogeneity of the graph implies that for any triple $(v_1, v_2, v_3) \in \Pi_3$ the conditional probabilities
$
P\{ v_1 = A_i | (v_2, v_3) = (A, B) \}
$
do not depend of the choice of particular triple $(v_1, v_2, v_3)$ and they are only a function of $A_i, A$ and $B$. Denote their common value
\[
P\{ v_1 = A_i | (v_2, v_3) = (A, B) \} = p_{ A_i | A - B }.
\]
Now, we compute from equation \eqref{eq:linky} to obtain
\begin{align*}
[A_i A B] &= \sum_{(v_1, v_2, v_3) \in \Pi_3} P\{ v_1 = A_i | (v_2, v_3) = (A, B) \} \P\{ (v_2, v_3) = (A, B) \} \\
&= \sum_{(v_1, v_2, v_3) \in \Pi_3} p_{ A_i | A - B } \P\{ (v_2, v_3) = (A, B) \}\\
&= \sum_{( v_2, v_3) = (A, B)} \P\{ (v_2, v_3) = (A, B) \} \sum_{v_1: v_1 \sim v_2, v_1 \neq v_3} p_{ A_i | A - B }\\
&= \sum_{( v_2, v_3) = (A, B)} \E(\mathbbm1\{ (v_2, v_3) = (A, B) \}) (n-1) p_{ A_i | A - B }\\
&= [AB] (n-1) p_{ A_i | A - B }.
\end{align*}
Hence, in order to compute the closure $[A_i A B]$, one needs a good approximation for the probabilities  $p_{ A_i | A - B }$. 
%Suitable approximations for this, that also take into account the clustering of $\mathcal G$, are discussed in the sequence of this article.

Note that, again, under Assumption \ref{ass:m} of the multinomial link model of Section \ref{subsec:multi}, one can immediately see that
\[
p_{ A_i | A - B } = p_i^{(v_2, A)} = p_i^{(A)} =  \frac{[A_iA]}{n [A]},
\]
which leads to the same result as given in Theorem \ref{thm:closure} and in~\cite{Taylor2012JMB,barnard2019epidemic}. The multinomial link model does not take into account the extra information given by $B$, so we expect that more refined closures can be discovered by taking this information into account.

\subsection{Closures for the Uniform distribution of states of the neighbours}

We reiterate the arguments of Section \ref{subsec:multi} in an example where the conditional distribution of types in neighbouring nodes is uniform.
For simplicity we are assuming that the network is $n$-regular.

\begin{example}(Uniform link distribution around a node.) Consider $m$ possible states for a node on a network $\mathcal G$ that is $n$-regular. Furthermore, we know that given the state $A$ of a node $v$ there exists an integer number $k_{i,j}^{(A)}$, $0 \le k_{i,j}^{(A)} \le n$ the conditional distribution
\be\label{ass:mixing}
(N_i^{(v)} , N_j^{(v)}) | v = A \sim \text{Uniform}[\Delta_{k_{i,j}^{(A)}}],
\ee
where $\Delta_{\ell} = \{ (x_1, x_2):  x_i \in \Z_+, x_1+x_2 \le \ell \}$. Using this, and starting from equation \eqref{eq:more-ex}, we calculate
\begin{align*}
[A_iAA_j] &= \sum_{v \in \mathcal G} \P\{ v = A\} E \Big[ N_i^{(v)}N_j^{(v)} \Big| v= A\Big]\\
&=\sum_{v \in \mathcal G}    \P\{ v = A\} \sum_{x = 0}^{k_{i,j}^{(A)}}\sum_{y=0}^{k_{i,j}^{(A)}-x}xy P\{ N_i^{(v)} = x, N_j^{(v)}=y \}\\
&=\sum_{v \in \mathcal G}    \P\{ v = A\} \sum_{x = 0}^{k_{i,j}^{(A)}}\sum_{y=0}^{k_{i,j}^{(A)}-x}xy \frac{2}{(k_{i,j}^{(A)}+1)(k_{i,j}^{(A)}+2)}.
\end{align*}
The last line comes from the fact the conditional joint distribution is uniform on the simplex $\Delta_{k_{i,j}^{(A)}}$. Now, the two inner sums can be directly computed to be
\begin{align*}
 \sum_{x = 0}^{k_{i,j}^{(A)}}\sum_{y=0}^{k_{i,j}^{(A)}-x}xy %&=  \frac{1}{2}\sum_{x = 1}^{k_{i,j}^{(A)}}x (k_{i,j}^{(A)}-x)(k_{i,j}^{(A)}-x+1)\\
 &= \frac{1}{2}\sum_{x = 1}^{k_{i,j}^{(A)}} \big\{ x ((k_{i,j}^{(A)})^2+ k_{i,j}^{(A)})  - x^2 (2 k_{i,j}^{(A)}+1) + x^3\big\} \\
 &= \frac{1}{24} k_{i,j}^{(A)}(k_{i,j}^{(A)}+1)(k_{i,j}^{(A)}-1)(k_{i,j}^{(A)}+2).
\end{align*}
Substitute in the calculation for $[A_iAA_j]$ to obtain
\begin{align} \label{eq:ucl}
[A_iAA_j] &= \frac{1}{12}\sum_{v \in \mathcal G} \P\{ v = A\}  k_{i,j}^{(A)}(k_{i,j}^{(A)}-1) \approx \frac{[A]}{12} k_{i,j}^{(A)}(k_{i,j}^{(A)}-1).
\end{align}
The last approximation holds when $ \P\{ v = A\} \approx [A]N^{-1}$ and when the number $k_{i,j}^{(A)}$ does not depend on $v$. 
It may be necessary to approximate $k_{i,j}^{(A)}$ using the network; one such approximation can be obtained by taking expected values 
\begin{align*}
[AA_i] &= \sum_{v \in \mathcal G} \P\{ v = A \} E[N_i^{(v)}| v =A]  \approx  \sum_{v \in \mathcal G} \frac{[A]}{N} \frac{k_{i,j}^{(A)}}{2}= [A]  \frac{k_{i,j}^{(A)}}{2}.
\end{align*}
This gives that 
$
k_{i,j}^{(A)} \approx 2 \frac{[AA_j]}{[A]}.
$
A similar calculation would give that $k_{i,j}^{(A)} = 2 \frac{[AA_i]}{[A]}$. We can combine the two to obtain two different estimates for $k_{i,j}^{(A)}$, namely
\[
\hat k_{i,j}^{(A)} = \frac{2}{[A]} \sqrt{[AA_i][AA_j]}  \quad \text{ or } \quad \tilde k_{i,j}^{(A)} = \frac{1}{[A]}([AA_i] + [AA_j]). 
\]
A final substitution of these estimates in \eqref{eq:ucl} yields two different closures for $[A_iAA_j]$. In fact, there are infinitely many closures implied; if one introduces a parameter $\alpha \in [0,1]$ then any convex combination $\alpha \hat k_{i,j}^{(A)} + (1 - \alpha)  \tilde k_{i,j}^{(A)}$ is a different closure.
\end{example}

%After the example, we describe dynamics that lead precisely to this sort of behaviour in the simple case of two types $S$ and $I$.

%%%%%%%%%%%%%%%%%%%%%%%%%%%%%%%%%%%%%%%%%%%%%%%%%%%%
\section{Discussion}\label{sec:disc}
%%%%%%%%%%%%%%%%%%%%%%%%%%%%%%%%%%%%%%%%%%%%%%%%%%%%
A mean-field approximation to Markovian epidemics is widely used in various scientific disciplines. Its strength relies on a drastic reduction of the number of equations, which are a priori exponentially (in the number of nodes) many, but in the mean-field model are polynomially many. Usually, mean-field models are made possible by using some kind of closure which are approximations of higher-order moments by lower-order ones, e.g. approximating the expected number of triples by the expected number of singles and pairs. The earlier such approximations are performed the more information is lost. For example, closing the pairs is simpler than approximating triples but the accuracy of approximation will be worse. For the standard SIS and SIR epidemics one typically closes (or approximates) triples~\cite{rand1999correlation,keeling1999effects,KissMillerSimon}.

In this article we presented a top-down probabilistic approach to obtain a rigorous formula for the expected number of $k$-tuples (and in particular triples) in multi-type epidemics (Prop. \ref{prop:triples}). No assumptions were necessary for the calculations, so the proposition works in all types of networks. We then
proceeded by approximating the expected value of triples in various ways. Each method of approximating gave rise to a closure - either in a theoretical or a numerical/statistical way.

Our main contributions outside of the robust theoretical framework mentioned above, are summarised below:
\begin{enumerate}
\item Derivation of a ``moments" closure that performs well in SIS epidemics on regular (or networks with low degree heterogeneity), Erd\H{o}s-R\'enyi and bi-modal  graphs.
\item Derivation of closures already used in the literature. We provide a careful list of background assumptions that are necessary for each approximation to work, and which are sufficient to derive these classical closures.
\item Derivation of a few new closures which either bare similarities with existing ones or allowed us to obtain epidemic threshold results that are well known in the literature.
\item A derivation of the Poisson link model from the multinomial link model, that naturally arises for Erd\H{o}s-R\'enyi graphs of low connection probability.
\item Numerical verification of theoretical assumptions, particularly the marginal distribution of types of nodes around susceptible ones.
\end{enumerate}

The new closures show some promise but need some more testing and better understanding when these agree with known ones, on which network do they work best and in what way 
are they different from the existing ones. The new closure based on assuming a uniform distribution is a good example of how our method can be extended and applied beyond epidemic models. Finally, we hope that these new insight may stimulate more research and may lead to some more rigorous results for closures and closed systems in whatever context or application.

\end{document}